\documentclass[a4paper,UKenglish,cleveref, autoref, thm-restate]{lipics-v2021}
\nolinenumbers %

\pdfoutput=1 %
\hideLIPIcs  %

\graphicspath{{./fig/}}%

\title{Project-connex Decompositions and Tractability of Aggregate Group-by Conjunctive Queries} %

\titlerunning{Project-connex Decompositions and Tractability of Aggregate CQs} %

\author{Diego Figueira}
{Universit\'e de Bordeaux, CNRS, Bordeaux INP, LaBRI, France}
{diego.figueira@cnrs.fr}
{https://orcid.org/0000-0003-0114-2257}{} %

\author{Cibele Freire}{Computer Science Department, Bowdoin College, USA}
{cfreire@bowdoin.edu}
{https://orcid.org/0009-0000-4242-1311}{}

\authorrunning{D. Figueira and C. Freire} %

\Copyright{Diego Figueira and Cibele Freire} %



\EventEditors{John Q. Open and Joan R. Access}
\EventNoEds{2}
\EventLongTitle{42nd Conference on Very Important Topics (CVIT 2016)}
\EventShortTitle{CVIT 2016}
\EventAcronym{CVIT}
\EventYear{2016}
\EventDate{December 24--27, 2016}
\EventLocation{Little Whinging, United Kingdom}
\EventLogo{}
\SeriesVolume{XX}
\ArticleNo{XX}

\AtEndPreamble{%
  \newtheorem*{theorem*}{Theorem}
}

\usepackage{balance}
\usepackage{enumerate}
\usepackage{thm-restate}
\usepackage{hyperref}
\usepackage{cleveref}
\usepackage{stmaryrd}

\usepackage{amssymb} 
\usepackage{amsmath}
\usepackage{tikz-cd}
\usepackage{setspace}
\usepackage[bibliography=common]{apxproof} %
\usepackage{bbding} %

\usepackage{listings} %
\newcommand\Smallix{\fontsize{7.3}{7.5}\selectfont}
\NewCommandCopy{\proofqedsymbol}{\qedsymbol}%
\newcommand{\exampleqedsymbol}{{$\triangle$}}%
\AtBeginEnvironment{proof}{\renewcommand{\qedsymbol}{\proofqedsymbol}}
\AtBeginEnvironment{nestedproof}{\renewcommand{\qedsymbol}{\nestedproofqedsymbol}}
\AtBeginEnvironment{example}{%
  \pushQED{\qed}\renewcommand{\qedsymbol}{\exampleqedsymbol}%
}
\AtEndEnvironment{example}{\popQED\endexample}

\usepackage[xcolor, cleveref, notion, quotation, electronic]{knowledge}
\usepackage{mathcommand}
\knowledgeconfigure{protect quotation={tikzcd}}
\knowledgeconfigure{diagnose line=true, diagnose bar=true}

\definecolor{Dark Ruby Red}{HTML}{580507}
\definecolor{Dark Blue Sapphire}{HTML}{053641}
\definecolor{Dark Gamboge}{HTML}{be7c00}

\IfKnowledgePaperModeTF{
}{
    \knowledgestyle{intro notion}{color={Dark Ruby Red}, emphasize}
    \knowledgestyle{notion}{color={Dark Blue Sapphire}}
    \hypersetup{
        colorlinks=true,
        breaklinks=true,
        linkcolor={cGreen}, %
        citecolor={cGreen}, %
        filecolor={cGreen}, %
        urlcolor={cGreen},
    }
    \IfKnowledgeElectronicModeTF{
    }{
        \knowledgeconfigure{anchor point color={Dark Ruby Red}, anchor point shape=corner}
        \knowledgestyle{intro unknown}{color={Dark Gamboge}, emphasize}
        \knowledgestyle{intro unknown cont}{color={Dark Gamboge}, emphasize}
        \knowledgestyle{kl unknown}{color={Dark Gamboge}}
        \knowledgestyle{kl unknown cont}{color={Dark Gamboge}}
    }
}

\knowledge{text={t.f.a.e.}, italic}
  | tfae

\knowledge{text={i.e.}, italic}
  | ie

\knowledge{text={I.e.}, italic}
  | Ie

\knowledge{text={s.t.}, italic}
  | st

\knowledge{text={e.g.}, italic}
  | eg

\knowledge{text={vs.}, italic}
  | vs

\knowledge{text={w.r.t.}, italic}
  | wrt

\knowledge{text={a.k.a.}, italic}
  | aka

\knowledge{text={w.l.o.g.}, italic}
  | wlog
  
\knowledge{text={W.l.o.g.}, italic}
  | Wlog

\knowledge{text={cf.}, italic}
  | cf

\knowledge{text={iff}, italic}
  | iff

\knowledge{text={r.e.}, italic}
  | r.e.
  | re

\knowledge{wrap=\textsf}
  | NL

\knowledge{notion, text={paraNL}, wrap=\textsf}
  | para-NL
  | paraNL

\knowledge{url={https://en.wikipedia.org/wiki/Parameterized\_complexity\#FPT}, text={\textup{FPT}}, wrap=\textsf}
  | FPT

\knowledge{notion}
  | fixed-parameter tractable

\knowledge{text={ExpSpace}, wrap=\textsf}
  | EXPSPACE
  | ExpSpace

\knowledge{text={2ExpSpace}, wrap=\textsf}
  | 2EXPSPACE
  | 2ExpSpace

\knowledge{text={PSpace}, wrap=\textsf}
  | PSpace
  | PSPACE

  \knowledge{text={PTime}, wrap=\textsf}
  | PTime
  | ptime

\knowledge{text={FP}, wrap=\textsf}
  | FP

\knowledge{text={NP}, wrap=\textsf}
  | NP

  \knowledge{text={P}, wrap=\textsf}
  | P

  \knowledge{url={https://en.wikipedia.org/wiki/Parameterized\_complexity\#W[1]}, text={\ensuremath{\textup{W}[1]}}, wrap=\textsf}
  | W1

  \knowledge{text={\ensuremath{\# \textsf{P}}}, wrap=\textsf}
  | shP

  \knowledge{text={\ensuremath{\# {\cdot}\textsf{NP}}}, wrap=\textsf}
  | shNP

  \knowledge{text={\ensuremath{\textsf{P}^{\# \textsf{P}}}}, wrap=\textsf}
  | PshP
  
\knowledge{text={\ensuremath{\textsf{FP}^{\# \textsf{P}}}}, wrap=\textsf}
  | FPshP

\knowledge{text={\ensuremath{\Pi^p_2}}, wrap=\textsf}
  | PiP2
  | Pi2
  | pi2
  | Pitwo
  | PiTwo
  | pitwo

\knowledge{url={https://en.wikipedia.org/wiki/Savitch\%27s_theorem}}
  | Savitch's Theorem

\knowledge{notion}
 | notion@notice
 | definition@notice

\knowledge{notion}
 | conjunctive queries
 | CQs
 | CQ
 | conjunctive query

\knowledge{notion}
 | $\aK $-CQ
 | $\aK $-CQs

 \knowledge{text={\ensuremath{\text{CQ}^{\textit{neg}}}}}
 | CQneg

\knowledge{notion}
 | subqueries

\knowledge{notion}
 | evaluation

\knowledge{notion}
 | self-join free
 | sjf
 | self-join

\knowledge{notion}
 | full
 | full CQ
 | full CQs

\knowledge{notion}
 | arity

\knowledge{notion}
 | tree-width
\knowledge{notion}
 | generalized hyperwidth

\knowledge{notion}
 | core

\knowledge{notion}
 | free atom
 | free atoms

\knowledge{notion}
 | $\bar x$-component

\knowledge{notion}
 | contraction
 | contractions

\knowledge{notion}
 | succinct representation

\knowledge{notion}
 | out-width

\knowledge{notion}
 | marginal out-width

\knowledge{notion}
 | marginal generalized hyperwidth

\knowledge{notion}
 | generalized hypertree decomposition
 | tree decomposition
 | tree decompositions

\knowledge{notion}
 | width

\knowledge{notion}
 | generalized hyperwidth

\knowledge{notion}
 | free connex
 | Free connex

\knowledge{notion}
 | free generalized hyperwidth

\knowledge{notion}
 | empty tuple

\knowledge{notion}
 | max-group-count

\knowledge{notion}
 | marginal free generalized hyperwidth
 | $\ell $-marginal free generalized hyperwidth

\knowledge{notion}
 | tuple aggregator
 | tuple aggregators

\knowledge{notion}
 | max-group-aggregation

\knowledge{notion}
 | semiring
 | semirings

\knowledge{notion}
 | ordered tuple aggregator

\knowledge{notion}
 | constants

\knowledge{notion}
 | variables
 | variable

\knowledge{notion}
 | terms

\knowledge{notion}
 | (relational) schema
 | schema

\knowledge{notion}
 | relation name
 | relation names

\knowledge{notion}
 | fact
 | facts

\knowledge{notion}
 | (relational) atom
 | atom
 | atoms

\knowledge{notion}
 | atomic

\knowledge{notion}
 | database
 | databases

\knowledge{notion}
 | homomorphism
 | homomorphisms
 \knowledge{notion}
 | equivalent

\knowledge{notion}
 | $C$-homomorphism
 | $C$-homomorphisms

\knowledge{notion}
 | free variables
 | free variable
 | free

\knowledge{notion}
 | bound
 | bound variables
 | bound variable

\knowledge{notion}
 | primal graph

\knowledge{notion}
 | bound path
 | bound paths
 | Bound paths

\knowledge{notion}
 | connectivity condition
 | Connectivity condition
 | connected@tdec
 | connectivity@tdec

\knowledge{notion}
 | bags
 | bag

\knowledge{notion}
 | contains@bag
 | containing@bag

\knowledge{notion}
 | acyclic
 | acyclicity

\knowledge{notion}
 | hypertree decomposition
 | hypertree decompositions

\knowledge{notion}
 | evaluation problem

\knowledge{notion}
 | underlying hypergraph
 | underlying hypergraphs

 \knowledge{notion}
 | hypergraph
 | hypergraphs

 \knowledge{notion}
 | graph
 | graphs

\knowledge{notion}
 | hyperwidth

\knowledge{notion}
 | Random Access Machines
 | RAM

\knowledge{notion}
 | Yannakakis algorithm

\knowledge{notion}
 | semi-joins
 | semi-join

 \knowledge{notion}
 | $\aK$-semi-join

\knowledge{notion}
 | completeness condition
 | Completeness condition

\knowledge{notion}
 | covering condition
 | Covering condition

\knowledge{notion}
 | assignment

\knowledge{notion}
 | singleton free connex

\knowledge{notion}
 | Boolean

\knowledge{url={https://en.wikipedia.org/wiki/Monoid\#Commutative_monoid}}
  | commutative monoid

\knowledge{url={https://en.wikipedia.org/wiki/Monoid}}
  | monoid

\knowledge{url={https://en.wikipedia.org/wiki/Tropical\_semiring}}
  | tropical

\knowledge{url={https://en.wikipedia.org/wiki/Polynomial-time\_approximation\_scheme\#Randomized}}
  | FPRAS

  \knowledge{url={https://en.wikipedia.org/wiki/Parameterized\_approximation\_algorithm}}
  | parameterized approximation schemes
  | parameterized approximation scheme
  
\knowledge{notion}
 | $C$-isomorphism
 | ($\const (q)$)-isomorphisms

\knowledge{notion}
 | query
 | queries

\knowledge{notion}
 | evaluation task

\knowledge{notion}
 | atom labeling

\knowledge{notion}
 | dimension

\knowledge{notion}
 | witnessing fact

\knowledge{notion}
 | $\Nat$-aggregate query
 | $\Nat$-aggregate queries
 | $K$-aggregate queries
 | $K$-aggregate query
 | $F$-aggregate query
 | $K$-aggregate query on a set $X$ of free variables
 | aggregate query
 | aggregate queries
 | $K$-aggregate

\knowledge{notion}
 | $S$-connex
\knowledge{notion}
 | $(S_1, \dotsc , S_t)$-connex

\knowledge{notion}
 | commutative semiring
 | commutative semirings

\knowledge{notion}
 | $K$-annotated database
 | annotated database
 | $\Nat$-annotated database
 | $K$-annotated databases
 | $F$-annotated database
 | $k$-annotated database
 | $\Mfacts $-annotated database
 | $(\Nat \times \Mfacts )$-annotated database
 | $\Nat$-annotated databases

\knowledge{notion}
 | join tree
 | join trees

\knowledge{notion}
 | free-connex

\knowledge{notion}
 | semiring $K$-aggregate query
 | semiring $K$-aggregate queries
 | semiring $\Nat$-aggregate query
 | semiring aggregate query
 | semiring aggregate queries

\knowledge{notion}
 | project-connex
 | Project-connectedness

\knowledge{notion}
 | annotation
 | annotations
 | annotated

\knowledge{notion}
 | counting CQs
 | counting CQ

\knowledge{notion}
 | witness subtree for $X$
 | witness subtree for
 | witness subtrees for
 | witness subtree
 | witness subtrees

\knowledge{notion}
 | $1$-initialized

\knowledge{notion}
 | trivial semiring

\knowledge{notion}
 | witnessing bijection
\knowledge{notion}
 | constant delay enumeration after a polynomial pre-processing
 | constant delay enumeration after a polynomial $\+O\big (\maxSize {D}^k\big )$ pre-processing

\knowledge{notion}
 | data complexity

\knowledge{notion}
 | constant delay enumeration after a linear pre-processing

\knowledge{notion}
 | project-connex generalized hyperwidth

\knowledge{notion}
 | compatible queries
 | compatible
\knowledge{notion}
 | frontier-path
 | frontier-paths
 | $X_i$-frontier-path
 | $X$-frontier-path
 | $X$-frontier-paths

\knowledge{notion}
 | augmented query
\knowledge{notion}
 | $P$-clique
 | $\PRel$-clique
\knowledge{notion}
 | $\PRel$-clique property
\knowledge{notion}
 | frontier-path@CQ
 | frontier-paths@CQ

\knowledge{notion}
 | combined complexity

\knowledge{notion}
 | UCQ
 | UCQs

\knowledge{notion, text={\ensuremath{\#\text{CQC}}}}
 | shCQC

\knowledge{notion}
 |characteristic hypergraphs
 |characteristic hypergraph

\knowledge{notion}
 | relativizing
 | relativization
 | relativized

\knowledge{notion}
 | measure
\knowledge{notion}
 | monotone measure
 | monotone measures
 | measure monotonicity
\knowledge{notion}
 | $\mu$-width
 \knowledge{notion}
 | fractional hyperwidth
 \knowledge{notion}
 | submodular width
\knowledge{notion}
 | bounded arity
 | unbounded arity

\knowledge{notion}
 | promise problem
 | promise problems

\usepackage{dutchcal} %
\usepackage{xspace} %
\usepackage{booktabs} %
\usepackage{bbding} %

\usepackage{caption}
\usepackage{subcaption}

\definecolor{Desire}{HTML}{eb3b5a} %
\definecolor{Boyzone}{HTML}{2d98da} %
\definecolor{Royal Blue}{HTML}{3867d6} %
\definecolor{NYC Taxi}{HTML}{f7b731} %
\definecolor{Beniukon Orange}{HTML}{fa8231}
\definecolor{Algal Fuel}{HTML}{20bf6b} %
\definecolor{A darker green}{HTML}{145a32} %
\definecolor{Innuendo}{HTML}{a5b1c2} %
\definecolor{Twinkle Blue}{HTML}{d1d8e0} %
\definecolor{Blue Horizon}{HTML}{4b6584} %
\definecolor{Gloomy Purple}{HTML}{8854d0} %

\colorlet{cBlue}{Royal Blue}
\colorlet{cYellow}{NYC Taxi}
\colorlet{cOrange}{Beniukon Orange}
\colorlet{cGreen}{A darker green}
\colorlet{cRed}{Desire}
\colorlet{cGrey}{Innuendo}
\colorlet{cDarkGrey}{Blue Horizon}
\colorlet{cLightGrey}{Twinkle Blue}
\colorlet{cPurple}{Gloomy Purple}

\renewcommand{\epsilon}{\varepsilon}

\usepackage[disable, backgroundcolor=orange!20, textcolor={Dark Ruby Red}, textsize=tiny]{todonotes}

\definecolor{green}{RGB}{0,120,0}
\definecolor{hlyellow}{RGB}{250, 250, 190}
\definecolor{diegoeditcolor}{RGB}{210,210,255}
\definecolor{cibeleeditcolor}{RGB}{210,255,210}
\newcommand{\sidediego}[1]{}
\newcommand{\sidecibele}[1]{}
\newcommand{\cibele}[1]{}

\newcommand{\diego}[1]{}
\definecolor{light-gray}{gray}{0.9}

\newcommand{\proofcase}[1]{\noindent\colorbox{cLightGrey}{#1}~~}

\renewcommand{\phi}{\varphi}
\renewcommand{\leq}{\leqslant}
\renewcommand{\geq}{\geqslant}
\renewcommand{\emptyset}{\varnothing}

\knowledgenewrobustcmd{\dcup}{\mathop{\cmdkl{\uplus}}} %

\newcommand{\set}[1]{\{#1\}}
\newcommand{\mset}[1]{{\{}\!\!{\{}#1{\}}\!\!{\}}}
\newrobustcmd{\defeq}{\mathrel{\hat{=}}}
\newcommand{\eqdef}{\defeq}

\newrobustcmd{\Nat}{\mathbb{N}}
\newrobustcmd{\Reals}{\mathbb{R}}
\newrobustcmd\pset[1]{\+P(#1)} %
\newrobustcmd{\poly}{\mathop{\textrm{poly}}} %

\knowledgenewrobustcmd{\polyrx}{\mathrel{\cmdkl{\le_{\textit{poly}}}}} %
\knowledgenewrobustcmd{\polyeq}{\mathrel{\cmdkl{\equiv_{\textit{poly}}}}} %

\knowledgenewrobustcmd{\homto}{\mathrel{\cmdkl{\xrightarrow{\smash{\textit{hom}}}}}}
\knowledgenewrobustcmd{\Chomto}[1][C]{\mathrel{\cmdkl{\xrightarrow{\smash{{#1}\textit{-hom}}}}}}

\knowledgenewrobustcmd{\Const}{\cmdkl{\textup{Const}}}
\knowledgenewrobustcmd{\Var}{\cmdkl{\textup{Var}}}
\knowledgenewrobustcmd{\ann}{\cmdkl{\nu}}
\knowledgenewrobustcmd{\oneann}{\cmdkl{\mathbf{1}}}

\knowledgenewrobustcmd{\atoms}{\cmdkl{\textit{atoms}}}
\knowledgenewrobustcmd{\vars}{\cmdkl{\textit{vars}}}
\knowledgenewrobustcmd{\const}{\cmdkl{\textit{const}}}
\knowledgenewrobustcmd{\mterms}{\cmdkl{\textit{term}}}
\knowledgenewrobustcmd{\arity}{\cmdkl{\textup{arity}}}
\knowledgenewrobustcmd{\class}{\mathcal{C}}
\knowledgenewrobustcmd{\evalgD}[2]{#1\cmdkl{\langle}#2\cmdkl{\rangle}} %
\knowledgenewrobustcmd{\evalaggD}[2]{#1\cmdkl{\langle}#2\cmdkl{\rangle}}
\knowledgenewrobustcmd{\evalqD}[2]{#1\cmdkl{\langle}#2\cmdkl{\rangle}}
\knowledgenewrobustcmd{\assto}{\mathrel{\cmdkl{\mapsto}}}

\knowledgenewrobustcmd{\core}{\cmdkl{\textit{core}}}
\knowledgenewrobustcmd{\emptytup}{\cmdkl{()}}
\knowledgenewrobustcmd{\hyperq}[1][q]{\cmdkl{\mathbf{G}_{#1}}}
\knowledgenewrobustcmd{\primalq}[1][q]{\cmdkl{\mathbf{G}^p_{#1}}}
\knowledgenewrobustcmd{\dimtup}{\cmdkl{\dim}}
\knowledgenewrobustcmd\vertex[1]{\cmdkl{V}(#1)}
\knowledgenewrobustcmd\edges[1]{\cmdkl{E}(#1)}

\knowledgenewrobustcmd{\maxSize}[1]{\cmdkl{\|}#1\cmdkl{\|}}
\knowledgenewrobustcmd{\sizeofD}[1][D]{\cmdkl{|}#1\cmdkl{|}}
\knowledgenewrobustcmd{\normOne}[1]{\cmdkl{\|}#1\cmdkl{\|_1}}
\knowledgenewrobustcmd{\normInf}[1]{\cmdkl{\|}#1\cmdkl{\|_{\infty}}}

\knowledgenewrobustcmd{\sizeofq}[1][q]{\cmdkl{|}#1\cmdkl{|}}
\knowledgenewrobustcmd{\sizeofgamma}[1][\gamma]{\cmdkl{|}#1\cmdkl{|}}

\knowledgenewrobustcmd{\reducesto}{\mathrel{\cmdkl{\leq_{\textit{poly}}}}}

\knowledgenewrobustcmd\bagmap{\cmdkl{\mathbf{b}}}
\knowledgenewrobustcmd\atommap{\cmdkl{\mathbf{a}}}
\knowledgenewrobustcmd\atommaplab{\cmdkl{\mathbf{\tilde a}}}
\knowledgenewrobustcmd\tagmap{\cmdkl{\mathbf{t}}}
\knowledgenewrobustcmd\tagmappath[1]{\cmdkl{\mathbf{t}[#1]}}
\newrobustcmd\tagmappathprime[1]{%
  \withkl{\kl[\tagmappath]}{%
    \cmdkl{\mathbf{t}'[#1]}%
  }%
}
\knowledgenewrobustcmd{\ghw}{\cmdkl{\textit{ghw}}}
\knowledgenewrobustcmd{\fghw}{\cmdkl{\textit{fghw}}}
\knowledgenewrobustcmd{\pghw}{\cmdkl{\textit{pghw}}}
\knowledgenewrobustcmd{\hw}{\cmdkl{\textit{hw}}}
\knowledgenewrobustcmd{\tw}{\cmdkl{\textit{tw}}}

\knowledgenewrobustcmd{\sjoin}[1][]{\mathop{\cmdkl{\ltimes_{#1}}}}
\knowledgenewrobustcmd{\costjoin}{\cmdkl{c_{sj}}}

\knowledgenewrobustcmd{\abotimes}[2]{\cmdkl{\bigotimes_{#1}}#2} %
\knowledgenewrobustcmd{\aboplus}[2]{\cmdkl{\bigoplus_{#1}}#2} %

\knowledgenewrobustcmd{\contr}{\cmdkl{\textit{contract}}}
\knowledgenewrobustcmd{\Mfacts}{\cmdkl{\mathbf{M}}}
\knowledgenewrobustcmd{\DBs}[1][\Sigma]{\cmdkl{\textup{DB}_{#1}}}
\knowledgenewrobustcmd{\evalPb}[1]{\cmdkl{\textup{\textsc{Eval-}}}#1}
\knowledgenewrobustcmd{\aug}[1]{{#1}^{\cmdkl{+}}}%
\knowledgenewrobustcmd{\restrictG}[2][G]{#1\cmdkl{[}{#2}\cmdkl{]}}
\knowledgenewrobustcmd{\evalCounting}[2]{\cmdkl{\#}{#1}\cmdkl{\langle}#2\cmdkl{\rangle}}
\knowledgenewrobustcmd{\counting}[1]{\cmdkl{\#}{#1}}
\knowledgenewrobustcmd{\PRel}{\cmdkl{P}}

\newcommand{\mysubparagraph}[1]{\textit{#1}~}

\newrobustcmd{\schemaName}[1]{\textsf{#1}}

\newtheorem{assumption}{Assumption}[section]
\newtheoremrep{proposition}{\sc Proposition}[section]
\newtheoremrep{lemma}{\sc Lemma}[section]
\newtheoremrep{theorem}{\sc Theorem}[section]
\newtheoremrep{claim}{\sc Claim}[section]
\Crefname{claim}{claim}{claims}
\Crefname{claim}{Claim}{Claims}
\newtheoremrep{corollary}{\sc Corollary}[section]
\AtBeginDocument{%
  \providecommand\BibTeX{{%
    \normalfont B\kern-0.5em{\scshape i\kern-0.25em b}\kern-0.8em\TeX}}}

\begin{document}

\maketitle

\begin{abstract}
	We introduce `project-connex' tree-width as a measure of tractability for counting and aggregate conjunctive queries over semirings with `group-by' projection (also known as `AJAR' or `FAQ' queries). 
This elementary measure allows to obtain comparable complexity bounds to the ones obtained by previous structural conditions tailored for efficient evaluation of semiring aggregate queries, enumeration algorithms of conjunctive queries, and tractability of counting answers to conjunctive queries.

Project-connex tree decompositions are defined as the natural extension of the known notion of `free-connex' decompositions. 
They allow for a unified, simple and intuitive algorithmic manipulation for evaluation of aggregate queries and explain some existing tractability results on conjunctive query enumeration, counting conjunctive query evaluation, and evaluation of semiring aggregate queries. 
Using this measure we also recover results relating tractable classes of counting conjunctive queries and bounded free-connex tree-width, or the constant-time delay enumeration of semiring aggregate queries over bounded project-connex classes.
We further show that project-connex tree decompositions can be obtained via algorithms for computing classical tree decompositions.

\end{abstract}

\noindent
\raisebox{-.4ex}{\HandRight}\ \ This pdf contains internal links: clicking on a "notion@@notice" leads to its \AP ""definition@@notice"".%

\section{Introduction}
\label{sec:intro}

    The aim of this work is to offer an alternative structural measure for the study of evaluation and enumeratimacon of aggregate conjunctive queries. We introduce \textbf{project-connex width} 

    as a natural and intuitive extension of the notion of `free-connex' width having an intuitive definition.
Our approach views aggregation as a generalized form of projection over suitable semirings in conjunctive queries, in line with 
the seminal work on the AJAR \cite{JoglekarPR16} and FAQ frameworks \cite{KhamisNR16, 
DBLP:journals/tods/KhamisCMNNOS20,
khamis2023faqquestionsaskedfrequently}.

We adopt the widely used \emph{annotated database} setting \cite{GreenKT07}, where each tuple in a database is associated with an annotation drawn from a given commutative semiring. Annotations can be combined, depending on the context, using the semiring's  multiplicative ($\otimes$) and additive ($\oplus$) operations.
In this setting, the relational algebra \emph{join} ($\Join$) between two relations results in a new relation where annotations are combined by $\otimes$-multiplying the corresponding tuples.
On the other hand, the \emph{projection} ($\pi$) collapses all tuples that agree on a set of attributes, and their annotations are aggregated using the semiring's $\oplus$-additive operation.

    This semiring-based semantics supports efficient evaluation of conjunctive queries via Yannakakis-style algorithms on tree decompositions, even in the generalized setting.

We further consider the framework of \cite{JoglekarPR16,KhamisNR16}, in which
projections over different sorts of `sums' can be nested, giving the possibility of nesting `group-by' operations on conjunctive queries. We give below some examples to build intuition on this kind of queries.

\subparagraph*{An introductive example}
Consider the query asking to return each salesperson in Boston together with their maximum sales produced in any given day, over a database with schema $\schemaName{Branch(city,person)}$, $\schemaName{Sale(person,date,product,qty)}$, $\schemaName{Price(product,amount)}$, where each row in a database is \emph{annotated with a number}: the last number of the row in the case of \schemaName{Sale} and \schemaName{Price} tables, and `1' in the case of \schemaName{Branch}, as depicted in \Cref{fig:ex-intro}-(a).
  \begin{figure}
    \includegraphics[width=1\textwidth]{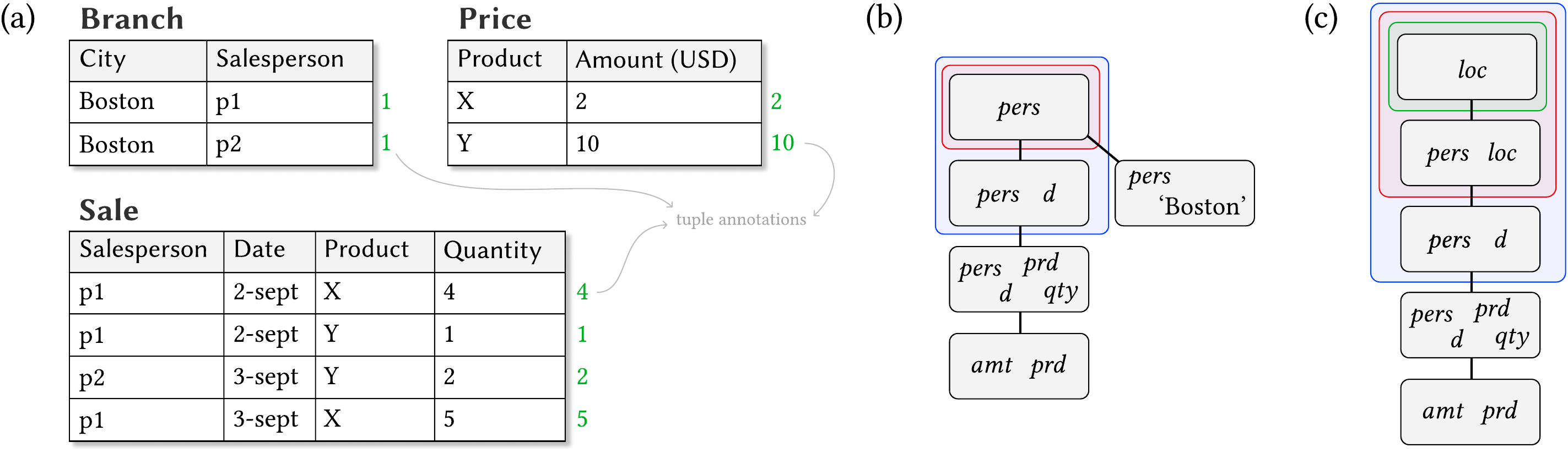}
    \caption{(a) A "database" $D$ of sales, where row annotations are in green. (b), (c) "project-connex" tree decompositions.
    The evaluation $Q(D)$ of $Q$ on $D$ outputs salesperson `p1' with value $18$ and `p2' with value  $20$; while
    $Q'(D)$ outputs `Boston' with value $38$.}
    \label{fig:ex-intro}
  \end{figure}
  To compute the query's output, we need a natural join where the annotations are \emph{multiplied} when joined. We then use a join parameterized by the multiplication operation on numbers, "ie",
  $
   \schemaName{Branch}(\textup{`Boston'},\textit{pers})
    \Join^{\textup{mult}}  
  \schemaName{Sale}(\textit{pers},d,\textit{prd},qty)
    \Join^{\textup{mult}}  
  \schemaName{Price}(\textit{prd},amt)
  $,
  that we just write as
  $\Join^{\textup{mult}}\set{\schemaName{Branch}(\textup{`Boston'},\textit{pers}),$ 
  $\schemaName{Sale}(\textit{pers},d,\textit{prd},qty),$ 
  $\schemaName{Price}(\textit{prd},amt)}$.
  Now we need to aggregate: first by summing all sales of each day/person, and then findinig the maximum for each person.
  This is done using projection operators, which specify an aggregation operation to perform for rows having the same projection. 
  Then the "aggregate query" in the language we study is the following:
  \[
    Q \defeq \pi^{\max{}}_{\set{\textit{pers}}} \pi^{\textup{sum}}_{\set{\textit{pers},d}} \Join^{\textup{mult}}\set{\schemaName{Branch}(\textup{`Boston'},\textit{pers}), 
  \schemaName{Sale}(\textit{pers},d,\textit{prd},qty), 
  \schemaName{Price}(\textit{prd},amt)}.
  \]
  When can this kind of queries be evaluated efficiently? Observe that both $(\Nat,\textup{sum},\text{mult})$ and $(\Nat,\max{},\text{mult})$ are "commutative semirings", and it is in this context that we work: queries in which the join-operation and the project-operation form a "semiring". 
  When working over a "commutative semirings", efficient evaluation methods based on tree decompositions can often be adapted exploiting the versatility offered by the "semiring" distributivity.

  To understand whether the evaluation can be done efficiently, one may consider `tree decompositions', such as the decomposition of $Q$ of \Cref{fig:ex-intro}-(b). 
  Observe how each bag ("ie", each gray potato of variables) is covered%
        \footnote{`Covered' in the sense that the variables of the atom contain all variables of the bag. Generally, it would be covered by $k$ atoms if there were $k$ atoms whose variables contain all the bag's variables.} 
  by \emph{one} atom of the query and hence its `width' is said to be \emph{one} (this is known as the `generalized hypertree width', or `coverwidth'). 
  Further, there are two subtrees (highlighted in red and in blue) which: (i) contain precisely the variables of the projections ("ie", $\set{\textit{pers}}$ and $\set{\textit{pers},d}$ resp.), and (ii) the first one is contained in the second one. We shall say that such decomposition is then `"project-connex"', since each projection set has an associated connected subtree. As we show, the existence of this decomposition witnesses an efficient evaluation. 
  In this example, since the width is one, the query admits a "constant delay enumeration after a linear pre-processing" (\Cref{thm:complexity-aggregate-acyclic-evaluation}). If the query was numeric, for instance by outputting the maximum sales by a person in any given day (which would be formally denoted by $\pi^{\max{}}_{\emptyset} Q$), then the output can be computed in linear time (\Cref{lem:evalAggQuery:width1}). And more generally, the output can be computed in $\+O(|Q| \cdot n^k)$ time for project-connex width $k$ queries, where $n$ is the size of the database (\Cref{thm:complexity-aggregate-evaluation}).
  This is true also for aggregate queries with an arbitrary nesting of aggregate operations.
  If instead of `Boston' we would prefer outputting, for each location, the sum of maximum sales of any salesperson in a day ("eg", as a bound to the maximum future turnover in a day per location), we could have written
  \[
    Q' \defeq \pi^{\textup{sum}}_{\set{\textit{loc}}} \pi^{\max{}}_{\set{\textit{pers},\textit{loc}}} \pi^{\textup{sum}}_{\set{\textit{pers},d,\textit{loc}}} \Join^{\textup{mult}}\set{\schemaName{Branch}(\textit{loc},p), 
  \schemaName{Sale}(\textit{pers},d,\textit{prd},qty), 
  \schemaName{Price}(\textit{prd},amt)}
  \]
  and there would still be the "project-connex" decomposition \Cref{fig:ex-intro}-(c) witnessing tractability.

\medskip

Depending on the number of nested projections and choice of semirings on these aggregate queries, we may find different previously studied query languages. For example:
\begin{itemize}
    \item 
    If there is only one projection which covers all variables and we work on the $(\Nat,\text{sum},\text{mult})$ semiring, the queries express counting homomorphisms from conjunctive queries (or, equivalently, counting the number of answers to full conjunctive queries),
    and its tractability has been characterized as having bounded tree-width \cite{DalmauJ04}.
    \item 
    If there is only one projection but which may not cover all variables, we may be rather interested in the efficient enumeration of answers;  in this vein we find the setting of enumerating answers to conjunctive queries, for which bounded `free-connex' tree-width has been identified as a fundamental measure  in the pioneering work \cite{BaganDG07}.
    \item 
    On the other hand, a particular case of two nested projections is counting the number of answers to conjunctive queries, for which other measures characterizing tractability have been developed \cite{DurandMengel15}.
\end{itemize}

Works on constant-delay enumeration, counting, and provenance have initially evolved rather independently, proposing a series of tractability measures on queries with two different sorts of problems in mind: constant-delay enumeration and direct access on the one hand, and evaluation on the other.
Some structural query measures found in the literature are:
\begin{itemize}
    \item `free-connex' tree decompositions \cite{BaganDG07} in the context of Conjunctive Query enumeration;
    \item `quantified star-size' \cite{DurandMengel15}, tree-width of `contracted' queries \cite{DBLP:conf/icdt/ChenM15}, or `$\#$-hypertree decompositions' \cite{GrecoS14} in the context of  counting solutions to Conjunctive Queries;
    \item `decomposable' and `valid' tree decompositions \cite{JoglekarPR16} and `FAQ-width' \cite{KhamisNR16,khamis2023faqquestionsaskedfrequently} in the context of evaluation of aggregate Conjunctive Queries over arbitrary semirings.
\end{itemize}

\subparagraph*{Contributions}
We show how a \textit{simple measure} like project-connex width can explain most tractability results in the areas above, allowing an elementary and uniform algorithmic treatment for both evaluation and enumeration. 
Hence, this work is guided by the question of whether there is a query measure which is simultaneously
\emph{simple to define}, \emph{easy to compute}, and \emph{admitting simple algorithmic manipulation} that can be used to \emph{give an account of the basic tractability results on aggregate conjunctive queries, including enumeration}.
This quest leads to the proposal of the project-connex measure, that we believe can improve the understanding and teaching of some of the rich and deep results developed in the literature.%

It is important to remark that influential previous work \cite{KhamisNR16,khamis2023faqquestionsaskedfrequently} already unifies known results for CQs and $\#$CQs (and more), but our approaches differ on methodology and scope (more on this in \Cref{sec:relworks}).

We adopt  the general approach for aggregation, provenance and counting, in such a way that they can
be expressed and intermixed in a rather intuitive way, in queries which also admit free variables and are hence amenable to enumeration algorithms. Our "aggregate queries" are in line with the FAQ/AJAR approach\footnote{There are some syntactical differences on how the language is defined, but the underlying concept is similar.} stemming from the seminal works \cite{KhamisNR16,khamis2023faqquestionsaskedfrequently, JoglekarPR16}, 
which in particular includes a form of group-by queries. This allows to capture queries such as taking the maximum of aggregate queries, or enumerating the tuples satisfying a query together with some aggregate value, or more complex nested queries.
For example, it is often the case that efficiently-enumerable queries are also efficient for counting the number of answers \cite{berkholz2020constant}. Here we see `enumeration' and `counting' as two sorts of queries from a same ``"aggregate query language@aggregate query"'', with a \emph{unified evaluation algorithm}. We give next some more detail on our contributions.

\smallskip

\mysubparagraph{Introducing structural conditions for tractability.}
We generalize the notion of ``free-connex'' tree decompositions to ``project-connex'', where we allow an arbitrary nesting of projections (\Cref{sec:project-connex-intro}). 
This allows to obtain complexity upper bounds for nested group-by queries.

\smallskip

\mysubparagraph{Computing project-connex width.}
We show that the project-connex width of an aggregate query $\gamma$ can be simply computed as the tree-width of a conjunctive query $q$
which can be obtained from $\gamma$ in polynomial time. Further, we show that project-connex decompositions of width 1 ("ie", acyclic) can be produced in linear time (\Cref{para:comput-tree-dec,para:comp-acyclic-dec}).

\smallskip

\mysubparagraph{Providing combined complexity bounds.}
In \Cref{sec:evalAgg} we provide complexity bounds via an elementary algorithm which yields a modest improvement ---to the best of our knowledge--- in the state-of-the-art bounds for evaluation in combined complexity of counting and enumeration of CQs with bounded free-connex width. 
Further, the same algorithm also yields the constant delay enumeration of "aggregate queries", whenever there are free variables.

\smallskip

\mysubparagraph{Simplifying the tractability measure for counting CQs.} 
For the case of counting answers to conjunctive queries, we show that project-connex boundedness (which in the setting of CQs is just free-connex) is equivalent to the known characterization for tractability \cite{DurandMengel15}.
Arguably, this simplifies the tractability characterization theorem\footnote{It simplifies the statement, but not the proof, which crucially relies on prior characterization results by Chen, Durand, and Mengel for counting CQs \cite{DurandMengel15,DBLP:conf/icdt/ChenM15}.} for counting CQs (\Cref{cor:freeghw:charPtime-body}).
Hence, an
 added value of using a generalization of free-connex decompositions is to make explicit known connections between counting and enumeration.

\subparagraph*{Organization}
We start on \Cref{sec:relworks} with a discussion of related work. After some preliminaries in \Cref{sec:prelim}, we introduce the "aggregate queries" we will study in \Cref{sec:aggregation}. 
In \Cref{sec:project-connex-intro}, we introduce the notion of "project-connex" decompositions and prove that computing such decompositions of any width is polynomial time, and linear time in the case of acyclic decompositions.
\Cref{sec:evalAgg} illustrates how decompositions are useful to algorithmic manipulation, by giving a simple procedure for evaluating and enumerating "aggregate queries".
\Cref{sec:countingcq} shows that project-connex decompositions can be used for characterizing tractability on counting CQs.
Finally, in \Cref{sec:conclusion} we give some concluding remarks. %
Due to space constraints, missing proof details and extended discussions are relegated to the appendix. %

\section{Related Work}\label{sec:relworks}

    \mysubparagraph{FAQ/AJAR framework.} Previous, seminal works \cite{KhamisNR16,JoglekarPR16,khamis2023faqquestionsaskedfrequently} study essentially the same aggregate query language as we do here. The emphasis of these works is put on `semantic equivalence', that is, characterizing when two queries, which only differ in the order in which the projections are applied, are equivalent depending on the semirings at hand. 
    In constrast, our measures are purely structural and completely blind to the specific semiring being used. For these reasons, our results do not necessarily imply any of the previous bounds of \cite{JoglekarPR16,KhamisNR16,khamis2023faqquestionsaskedfrequently}.
    Rather, we emphasize that bounded "project-connex" width can already explain tractability of some aggregate queries and leads to the elementary algorithm of \Cref{sec:evalAgg} for evaluation and enumeration.
    On the other hand, the cited works focus their analysis of query evaluation on data complexity.
    In comparison, our algorithmic analysis yields also combined complexity bounds.
    FAQ queries with additive inequalities (FAQ-AI) \cite{DBLP:journals/tods/KhamisCMNNOS20} were studied later, and their work shows better bounds for query evaluation than earlier results from \cite{KhamisNR16, khamis2023faqquestionsaskedfrequently}. Similar to our work and \cite{khamis2023faqquestionsaskedfrequently}, they obtain constant delay enumeration, when carefully considering how to represent the query output.

    On a syntactic level, the "aggregate queries" we introduce are slightly different to the ones of FAQ/AJAR since our projection operator projects \emph{onto a set of variables} while their projection operator \emph{projects away one variable} at a time. 
    Concretely, a CQ $q(\set{x_1, \dotsc, x_n})$, where $\vars(q)=\set{x_1, \dotsc, x_n, \dotsc, x_m}$ is expressed to something closer to $\pi^{\oplus}_{x_{n+1}} \pi^{\oplus}_{x_{n+2}} \dotsb \pi^{\oplus}_{x_m} {\Join^\otimes}q$ in their work instead of as $\pi^{\oplus}_{\set{x_1, \dotsc, x_n}} {\Join^\otimes}q$ in our formalism. 
    Among the reasons for our proposed syntax are to make the link to "free-connex" decompositions evident, since we project onto a set of variables, and to be able to give some complexity results in terms of the number of projections (\Cref{cor:lineartime-width1-pc-decompositions}), which implies a linear time algorithm for classes of bounded number of projections, independently of how many variables we may be projecting away.

    In addition, \cite{JoglekarPR16} also investigates structural conditions for efficient evaluation, based on a restriction of tree decompositions called ``decomposable''. 
     While decomposable and project-connex are distinct decomposition restrictions, they share some common aspects. Our objective in using project-connex decompositions lies in their intuitive connection to enumeration measures, allowing us to directly leverage known results.
    We also remark that \cite[Definition 3]{JoglekarPR16} proposes a width measure which may also bear some resemblances to project-width, as well as a way to compute it by means of reducing it to the computation of the tree-width of conjunctive queries (as we do here).\footnote{Unfortunately the characterization provided in the paper has an error ("cf"~\Cref{sec:problemAJAR}), but we believe that our characterization for project-width can be used instead, implying that the results of \cite{JoglekarPR16} relying on the characterization still hold.}

\smallskip

    \mysubparagraph{Counting.} \cite{DalmauJ04} offers the characterization of tractability for counting answers to \emph{"full"} "conjunctive queries" -- or, equivalently, counting "homomorphisms" -- based on the boundedness of tree-width.
    For arbitrary "conjunctive queries", the work \cite{DurandMengel15} initially established a characterization of tractability based on a measure of ``quantified star-size''.
    
    This was later expanded in \cite{DBLP:conf/icdt/ChenM15} and \cite{ChenGrecoMengelScarcello23}, where the proposed tractability characterization is based on the tree-width of two related queries: its core and the so-called `contraction' of the query.%
    \footnote{The need to use the tree-width measure on two separate queries is useful and probably necessary for deriving the trichotomy result of \cite{DBLP:conf/icdt/ChenM15} in \emph{parameterized complexity}.} As we show, if the focus is on obtaining a tractability characterization, one can just use the notion of width on tree decompositions which are ``free-connex''.

    Our bounds for the evaluation of counting queries are, to the best of our knowledge, in line or improving the algorithms in the literature. 
    In particular, \Cref{thm:CQ-bounded-fghw-counting-ptime} seems to improve on the bound of \cite[Theorem 5 with $k=0$]{PichlerS13} in combined complexity, which claims $\+O(\sizeofq \cdot \maxSize{D}^2)$ for "acyclic" "full" "CQs".

\smallskip

\mysubparagraph{Enumeration.}
The recent work \cite{EldarCarmeliKimelfeld23} studies the complexity of enumeration and direct-access to results of conjunctive queries with grouping and aggregation over annotated databases. The paper focuses on acyclic queries in relation to the \emph{direct-access problem} ("ie", the problem of accessing efficiently the $i$-th solution in the lex-ordering).
For self-join free queries the aggregate query language of \cite{EldarCarmeliKimelfeld23} can be translated into ours, with a different syntax ("cf" \Cref{app:enumeration-relwork}). For "self-join" queries the translation is not so simple, if at all possible.
\begin{toappendix}
    \mysubparagraph{Enumeration.} \label{app:enumeration-relwork} The recent work \cite{EldarCarmeliKimelfeld23} studies the complexity of enumeration and direct-access to results of conjunctive queries with grouping and aggregation over annotated databases. The paper focuses on acyclic queries in relation to the \emph{direct-access problem} ("ie", the problem of accessing efficiently the $i$-th solution in the lex-ordering).

For self-join free queries the aggregate query language of \cite{EldarCarmeliKimelfeld23} can be easily translated into ours, with a different syntax. 
For example, they use a syntax
$$
    Q(x_1,x_2,\textsf{Sum}(y_2),z) ~{{:}-}~ R(x_1,x_2,y_1), S(y_1,y_2,z)
$$
to express what in our formalism would be an "aggregate query" of the form
$$
    \pi^{+}_{\set{x_1,x_2,z}} \pi^{\max}_{\set{x_1,x_2,y_2,z}} {\Join^{\max}} \set{R(x_1,x_2,y_1), S(y_1,y_2,z)}
$$
on a "$\Nat$-annotated databases", "ie", "databases" where "facts" are "annotated" with numbers. In this case, $R$-"facts" are "annotated" with $0$ and each $S$-"fact" $S(n_1,n_2,n_3)$ is "annotated" with $n_2$. Intuitively, $\Join^{\max}$ performs the join by taking the maximum of "annotations", $\pi^{\max}$ projects away variables  by taking the maximum ("ie", preserving the "annotation"), and $\pi^{+}$ projects away variables by taking the sum ("ie", summing the second column of $S$). 
\end{toappendix}

\begin{toappendix}

\mysubparagraph{Provenance.}
The framework we present is compatible with (and inspired by) the seminal work on the semiring framework for database provenance by Green and Tannen \cite{GreenT17}, introduced to generalize the evaluation on database giving some information on the computation. Indeed "aggregate queries" are very close to the ``$K$-relational algebra'' on "commutative semirings" of \cite{GreenKT07} in the context of data provenance, with the key difference that the projection can take different sorts of sum-operations.
While we do not dwell here in exploring provenance, we remark that provenance semiring can be naturally ``plugged-in'' in the queries and databases we work with.
Moreover, from the standpoint of our framework the terms ``aggregation'' and ``provenance'' are equivalent: aggregation is a sort of provenance, and provenance is a sort of aggregation.
\end{toappendix}

\smallskip

\mysubparagraph{Output-Sensitive Evaluation.} 
The classic Yannakakis algorithm provides an output-sensitive query evaluation bound of $\+O(\sizeofD + \sizeofD \cdot |OUT|)$ for "acyclic" "CQs" \cite{Yannakakis81}, where $OUT$ is the size of the output.
This was the best known bound until very recently, when a series of papers made significant improvements.
In \cite{Hu24}, the author uses fast matrix multiplication to devise an algorithm that improves the bound for acyclic CQs, obtaining a running time of $\+O(\sizeofD + |OUT| +\sizeofD \cdot |OUT|^{5/6})$. Soon after, \cite{DeepZFK24} presents a generalization of Yannakakis algorithm to evaluate any acyclic CQ with running time 
$\+O(\sizeofD + |OUT| + \sizeofD \cdot |OUT|^{1-\varepsilon})$, where $0 < \varepsilon \leq 1$ is a constant that depends on the query.% 
They also show this bound is tight for star queries, and their algorithm is general enough to provide bounds for some aggregate queries (FAQ).
The most recent work improves on previous results via a semiring algorithm \cite{Hu25}. Hu proves an output-optimal bound of $\Theta(\sizeofD \cdot OUT^{1-1/\textit{fn-fhtw}} + OUT)$ for acyclic queries, where $\textit{fn-fhtw}$ is the free-connex fractional hypertree width of the query.

\section{Preliminaries}
\label{sec:prelim}

\AP 
We use the usual ``bar notation'' $\bar x$ for a tuple of elements, where ``$\intro*\emptytup$'' denotes the ""empty tuple"",  $\intro*\dimtup(\bar x)$ denotes the ""dimension"" of $\bar x$, and $\bar x[i]$ denotes its $i$-th element, for $i \in \set{1, \dotsc, \dimtup(\bar x)}$. In particular $\dimtup(\emptytup) = 0$.
\AP
We fix disjoint infinite sets 
$\intro*\Const$, 
$\intro*\Var$ of 
""constants"" and ""variables"", respectively. 
For any syntactic object $O$ ("eg" database, query, tuple), we will use $\intro*\vars(O)$ and $\intro*\const(O)$ 
to denote the sets of "variables" and "constants" contained in $O$, and let $\intro*\mterms(O)\eqdef \vars(O) \cup \const(O)$ denote its set of ""terms"".

\AP A ""(relational) schema"" is a finite set of relation symbols, each associated with an ""arity"" ($\geq 0$). %
\AP
A ""(relational) atom"" over a "schema" $\Sigma$ takes the form $R(\bar t)$ where $R$ is a ""relation name"" from $\Sigma$ of some arity $k$, and $\bar t \in (\Const \cup \Var)^k$.
\AP
A ""fact"" is an "atom" which contains only "constants".
\AP
A ""database"" $D$ over a "schema" $\Sigma$ is a finite set of "facts" over $\Sigma$. 
We denote by $\intro*\DBs$ the set of all "databases" over $\Sigma$, or $\reintro*\DBs[]$ if the "schema" is not relevant.
Let $\sizeofD$ be the size of the encoding  of $D$ and $\maxSize D$ be the maximum size of a relation of $D$, that is, 
$\intro*\sizeofD \eqdef \sum_{R \in \Sigma}\sum_{R(\bar c) \in D} \intro*\arity(R)$
and
$\intro*\maxSize D \eqdef \max_{R \in \Sigma}\sum_{R(\bar c) \in D} \arity(R)$.
\AP
A ""query"" is a computable function $\gamma : \DBs[] \to \+O$ for some output domain $\+O$. To improve readability we denote by $\intro*\evalgD{\gamma}{D}$ (instead of $\gamma(D)$) its evaluation on a "database" $D$. 
\AP
For any class $\class$ of "queries", the ""evaluation task"" $\intro*\evalPb{\class}$ is the task of, given a "database" $D$ and a "query" $\gamma \in \class$, computing $\evalgD{\gamma}{D}$.%
\footnote{
	   We understand this as a ``"promise problem"'', that is, the input to the evaluation problem is a query from $\+C$, as opposed to some arbitrary string.
       }
Unless otherwise stated, all the bounds we give for "evaluation" are in terms of ""combined complexity"", that is, where both the "query" and the "database" are taken as input (we assume any reasonable encoding under the "RAM" model).
\AP
We say $\class$ admits ""constant delay enumeration after a polynomial pre-processing"" if there is an algorithm which upon a "query" $\gamma$ and "database" $D$ first performs a pre-processing in $\+O(\sizeofD^k)$ for some $k$ in ""data complexity"" ("ie", where the query $\gamma$ is fixed) %
in which some data structure is created, and then it enumerates, one by one and without repetitions, all answers from 
$\evalqD{\gamma}{D}$ such that the time between two consecutive answers does not depend on the size of the "database". If $k=1$, we say that $\class$ admits ""constant delay enumeration after a linear pre-processing"". We work under a standard computational model of ""Random Access Machines"" (\reintro{RAM}) (more details in \Cref{app:RAM}).

\begin{toappendix}
    \begin{assumption}\AP\label{rk:promisepb}
        In the present manuscript we understand, for any class $\+C$ of "queries", the "evaluation task" $\evalPb{\class}$, as well as all other problems of this kind, as \AP`""promise problems""'. 
        That is, we assume that the input of the problem is a "query" from $\+C$ with a conventional encoding, rather than any string.
        This is rather standard and in line with previous works (as in, for example, all the foundational work regarding the characterization of tractable evaluation for CQs \cite{DBLP:journals/jacm/Grohe07}).
    \end{assumption}    
\end{toappendix}

\AP
For sets of "atoms" $S_1,S_2$, a ""homomorphism"" from $S_1$ to $S_2$ is a function  $h: \mterms(S_1) \rightarrow  \mterms(S_2)$  such that $R(h(t_1), \dotsc, h(t_k)) \in S_2$ for every  $R(t_1, \dotsc, t_k) \in S_1$. We write $S_1 \intro*\homto S_2$ to indicate the existence of such $h$.
\AP
If further we have that $h(c) = c$ for every $c \in C \subseteq \Const \cup \Var$, we call $h$ a ""$C$-homomorphism"" and write $S_1 \intro*\Chomto S_2$. If further $h$ is bijective and $h^{-1}$ is a "$C$-homomorphism", we call $h$ a ""$C$-isomorphism"".

\AP
A ""conjunctive query"" (henceforth just \reintro{CQ}) over $\Sigma$ is an expression of the form $\pi_X q$, where $q$ is a finite set of "atoms" over $\Sigma$, and $X$ 
is a subset of "variables" from $\vars(q)$. We will sometimes write $q(X)$ instead of $\pi_X q$. 
We call $X$ the ""free variables"" of $q(X)$; the remaining ones are ""bound"". If there are no "free variables" ("ie", $X=\emptyset$), then the query is ""Boolean"", and we often write it as $q()$. If all its variables are "free", then we call the query ""full"".
The size of an "atom" $\alpha$, noted $\intro*\sizeofq[\alpha]$, is the "arity" of the relation name it contains, and the size of a "CQ" $q(X)$, noted $\reintro*\sizeofq$, is the sum of sizes of its "atoms".
We use $\intro*\atoms(q)$ to denote the set of "atoms" of a "CQ" $q(X)$. A "CQ" is ""self-join free"" 
if it has no two "atoms" with the same "relation name". 
\AP
The "arity" of a class $\+C$ of "CQs", noted $\arity(\+C)$, is the maximum arity of the relations used in the queries of $\+C$, or $\infty$ otherwise. The  $\arity(\Sigma)$ of a "schema" is the maximum arity of a "relation name" therein.

\AP
The ""evaluation"" of a "CQ" $q(X)$ on a "database" $D$ is the set $\intro*\evalqD{q(X)}{D} \eqdef \set{h|_X  \mid h : q \Chomto D}$ for $C=\const(q)$, where $h|_X$ is the restriction of $h$ on $X$.
The ""evaluation problem"" for a class $\+C$ of "CQs" is the task of, given $q(X) \in \+C$ and $D$, computing $\evalqD{q(X)}{D}$. 
For a "Boolean" query $q()$, we often write $D \models q$ to denote $\evalqD{q()}{D} \neq \emptyset$.
\AP
Two "CQs" $q(X), q'(X)$ are ""equivalent"" if $\evalqD{q(X)}{D} = \evalqD{q'(X)}{D}$ for every "database" $D$ or, equivalently, if $q \Chomto q'$ and $q' \Chomto q$ for $C = X \cup \const(q) \cup \const(q')$.
\AP
The ""core"" of a "CQ" $q(X)$ is its minimal "equivalent" query, noted $\intro*\core(q(X))$, which is unique up to "($\const(q) \cup X$)-isomorphisms@$C$-isomorphism". The "core" of a class $\+C$ of "CQs" is $\reintro*\core(\+C) \eqdef \set{\core(q(X)) : q(X) \in \+C}$.
\AP
 A (partial) ""assignment"" is a (partial) function $f: \vars(q) \to \Const$. For $\bar x$ and $\bar c$ a tuple of variables from $q$ and of constants, respectively, of the same dimension, $\bar x \intro*\assto \bar c$ denotes the "assignment" mapping $\bar x[i]$ to $\bar c[i]$ for every $i$.\footnote{Assuming that there are no ``clashes'', "ie", that $\bar c[i] = \bar c[j]$ for every $i,j$ such that $\bar x[i] = \bar x[j]$.}%
\begin{toappendix}
\label{app:RAM}
\AP
As is usual for enumeration and counting problems, we will use ""Random Access Machines"" (\reintro{RAM}) whose domain contains $\Nat \dcup \Const \dcup \set{\bot}$ plus all the elements of the "semiring" under consideration, in which we assume that the "RAM"'s memory is initialized to $\bot$. 
For every fixed dimension $d \geq 1$  we have available an unbounded number of $d$-ary arrays $A$ such that for given $(n_1, \dotsc, n_d) \in \Const^d$ the entry $A[n_1, \dotsc, n_d]$ at position $(n_1, \dotsc, n_d)$ can be accessed in constant time.
Further, we assume unit-cost arithmetic operations, as well as semigroup sum and product operations.
Henceforward we will always use the "RAM" model of computation, unless otherwise stated.    
\end{toappendix}

\newcommand{\aK}{\mathbb{K}}
To obtain efficient evaluation algorithms, we focus on "aggregate queries" having a "semiring" embedded in its operators.%
\AP 
A ""semiring"" $\aK = (K, \oplus, \otimes )$ 
verifies:
  \begin{enumerate}[(i)]
    \item 
    $(K,\oplus)$ is a "commutative monoid" with identity $0_\oplus$;
    \item 
    $(K,\otimes)$ is a "monoid";
    \item 
    $\otimes$ distributes over $\oplus$ ("ie", $(a_1 \oplus \dotsb \oplus a_n) \otimes (b_1 \oplus \dotsb \oplus b_m) = \bigoplus_{i, j} a_i \otimes b_j$); and
    \item 
    $a \otimes 0_\oplus = 0_\oplus \otimes a = 0_\oplus$ for every $a \in K$.
  \end{enumerate}
It is a ""commutative semiring"" if $\otimes$ is further "commutative@commutative monoid". 

\paragraph*{Tree Decompositions}
\AP
We use the standard notation for ""hypergraphs"" $G$, where $\intro*\vertex G$ is the set of vertices and $\intro*\edges G$ the set of hyperedges. 
\AP
Given a "hypergraph" $G$ and a set of vertices $C$ thereof, let $\intro*\restrictG{C}$ be the subgraph of $G$ induced by $C$, obtained by replacing each edge $e \in \edges{G}$ with $e \cap C$ when the intersection is nonempty.
\AP
The ""underlying hypergraph""
$\intro*\hyperq$ of a  "CQ" $q(X)$ is the "hypergraph" where 
$V(\hyperq) = \vars(q)$
and 
$E(\hyperq) = \{\vars(\alpha) ~|~ \alpha \in \atoms(q)\}$.

\AP
A ""generalized hypertree decomposition"" (henceforth just \reintro{tree decomposition}) of a "hypergraph" $G$ is a tuple $(T,\intro*\bagmap, \intro*\atommap)$ where $T$ is a tree, $\bagmap : \vertex{T} \to \pset{\vertex G}$, and $\atommap: \vertex{T} \to \pset{\edges G}$ such that
\begin{enumerate}[(a)]
    \item for each hyperedge $e \in \edges G$, there is a vertex $v \in \vertex T$ such that $e \subseteq \bagmap(v)$ ---this is the ""completeness condition"";
    \item for each $v \in \vertex G$, the set of vertices $\set{t \in \vertex T : v \in \bagmap(t)}$ is connected in $T$ ---this is often known as the ""connectivity condition"";\footnote{If $v$ is isolated it may not appear in the decomposition. However, the "hypergraphs" we consider do not contain isolated vertices.}
    \item for each  $v \in \vertex T$, $\bagmap(v) \subseteq \bigcup\atommap(v)$ ---this is the ""covering condition"".
\end{enumerate}
\AP
We often refer to vertices of $T$ as ""bags"", and say that a "bag" $v$ ""contains@@bag"" some "variable" $x$ to denote $x \in \bagmap(v)$.
\AP
The ""generalized hyperwidth@width"" (henceforth just `\reintro{width}' for brevity) of such "tree decomposition" is the maximum cardinality of $\atommap(v)$, for $v$ ranging over all vertices. The ""generalized hyperwidth"" of a "hypergraph" $G$, or $\intro*\ghw(G)$, is the minimum "(generalized hyper-) width@width" among all its "tree decompositions". 
The \reintro{tree decomposition} and \reintro{generalized hyperwidth} of a "CQ" or any "$K$-aggregate query" is defined analogously by applying the definitions above to its "underlying hypergraph".
\begin{toappendix}
A known and useful fact, due to the "completeness@completeness condition" and "connectivity conditions@connectivity condition", is the following.
\begin{lemma}\label{lem:cliques-in-bags}
    If $G$ contains a clique\footnote{That is, $\set{v,v'} \in \edges{G}$ for every pair of distinct $v,v' \in X$.} on $X \subseteq \vertex{G}$, then every "tree decomposition" of $G$ has a "bag" "containing@@bag" $X$.
\end{lemma}
\AP
In the context of a "tree decomposition" $(T,\bagmap,\atommap)$ of a "CQ" %
$q(X)$ it is sometimes useful to have an explicit link between the "atoms" of $q$ and the hyperedges; we hence define an ""atom labeling"" as any mapping $\intro*\atommaplab  : \vertex{T} \to\pset{q}$ such that $\atommap(v) = \set{\vars(\alpha) : \alpha \in \tilde\atommap(v)}$ for every $v \in \vertex{T}$.
\end{toappendix}
\AP
An ""acyclic"" "CQ" or "aggregate query" is one having "generalized hyperwidth" $1$.
\AP
In this context, a ""join tree"" of a set $q$ of "atoms" is a "tree decomposition" of "width" 1 of $q$ together with a ""witnessing bijection"" $\phi: \vertex{T} \to q$ such that $\atommap(v) = \set{\vars(\phi(v))}$ and $\bagmap(v) = \vars(\phi(v))$ for every $v$. 

\begin{toappendix}
    The following are some useful known results on "tree decompositions".
\begin{lemma}\AP\hfill
    \AP\label{lem:jointree}
    \begin{enumerate}
        \item \AP\label{lem:jointree:width1=jointree}
        \cite[Theorem~4.5]{GottlobLS02} A "CQ" has a "join tree" if{f} it has a "tree decomposition" of "width" 1;%
        \item \AP\label{lem:jointree:lineartime}
        \cite{Graham1980universal,YuO79} Testing whether a "CQ" $q(X)$ has a "join tree" is decidable in $\+O(\sizeofq)$; if so, a "join tree" can be produced within the same time bounds.
    \end{enumerate}
\end{lemma}
\begin{lemma}[Consequence of \cite{ChenD05,DalmauKV02}]\AP\label{lem:equiv-sem-ghw-CQ}
    For any class $\+C$ of "CQs" such that $\ghw(\core(\+C))$ is bounded, checking whether two queries therein are "equivalent" is in polynomial time.
\end{lemma}
\begin{lemma}\AP\label{lem:core-sem-bound-ghw-polytime}
    For any class $\+C$ of "CQs" such that $\ghw(\core(\+C))$ is bounded, the "core" of a query from $\+C$ is polynomial-time computable.
\end{lemma}
\begin{proof}
    This is basically the same as \cite[Lemma~25]{DBLP:conf/icdt/ChenM15} generalized to "unbounded arity". Given $q(X) \in \+C$, we remove an atom from $q(X)$ obtaining $\bar q'(X)$, and we check if the queries are equivalent, which is in polynomial time due to \Cref{lem:equiv-sem-ghw-CQ}. We repeat until we cannot find such an atom, and it is well-known that this yields the "core" of $q(X)$ \cite{ChandraM77}.
\end{proof}
\end{toappendix}
\AP
The ""tree-width"" of a "tree decomposition" is the maximum cardinality of $\bagmap(v)$ minus one, and the "tree-width" of $q$, or $\intro*\tw(q)$, is the minimum "tree-width" among all its "tree decompositions". 
Notice that when working with "tree-width" we can disregard the $\atommap$-component from the definition of "tree decomposition" and simply describe a "tree decomposition" as a pair $(T,\bagmap)$.

\begin{toappendix}
    \begin{remark}\AP\label{rk:ghw-leq-tw}
        The "width" of any "tree decomposition" is at most its "tree-width"; hence, for every "CQ" $q$, $\ghw(q) \leq \tw(q)$.
        The "tree-width" of any "tree decomposition" is at most $k \cdot r -1$, where $k$ is its "width" and $r$ is the maximum "arity" of $q$.
    \end{remark}
\end{toappendix}

\begin{toappendix}
\AP
A "tree decomposition" $(T,\bagmap,\atommap)$ is a ""hypertree decomposition"" if 
$T$ is a rooted tree and
for every vertex $v \in \vertex T$ and descendant $v'$ thereof we have $\bagmap(v') \cap (\bigcup \atommap(v)) \subseteq \bagmap(v)$. 
The ""hyperwidth"" of a "hypergraph" $G$, or $\intro*\hw(G)$, is the minimum "width" among its "hypertree decompositions". 
The interest of such restriction is that it makes the problem of testing if a "CQ" has a "hypertree decomposition" of "width" $k$ in "PTime" (for any fixed $k$), whereas the problem is otherwise "NP"-complete \cite[Theorems 5.16 and 3.4 resp.]{GottlobLS02}.
\begin{theorem}[{\cite[Theorem 5.16]{GottlobLS02}}]\AP\label{lem:recognizability:hw:ptime}
    For any fixed $k$, the problem of testing if a "CQ" has a "hypertree decomposition" of "width" $k$ is in "PTime", if so, such decomposition can be produced in polynomial time.
\end{theorem}
\begin{theorem}[{\cite[Theorem 19-2 combined with Lemmas 15 and 16]{AdlerGG07}}]\AP\label{lem:ghw:hw:bound}
    For every "hypergraph" $G$, 
    $\ghw(G) \leq \hw(G) \leq 3 \cdot \ghw(G) + 1$.
\end{theorem}
\end{toappendix}

\subparagraph*{Free-Connex Decompositions}
\AP
Given a "tree decomposition" $(T,\bagmap,\atommap)$ of a "hypergraph" $G$ and a set $X\subseteq \vertex{G}$,  we say that $T'$ is a ""witness subtree for $X$"" (or simply a "witness subtree" if $X$ is clear from the context) if $T'$ is a connected subtree of $T$ such that $\bigcup_{v \in \vertex{T'}}\bagmap(v) = X$.
We say that a "tree decomposition" for a "CQ" $q(X)$ is ""free-connex"" if it contains a "witness subtree for $X$".
\AP
The ""free generalized hyperwidth"" of a "CQ" $q(X)$, or $\intro*\fghw(q(X))$, is the minimum "width" among all its "free-connex" "tree decompositions".

"Free connex" decompositions have been introduced and studied in the realm of enumeration algorithms for "CQs". In particular, "sjf" "CQs" of "free generalized hyperwidth" 1 ("aka" "free-connex" "acyclic" "sjf" "CQs") characterize (under complexity theoretic assumptions) the "sjf" "CQs" admitting a "constant delay enumeration after a linear pre-processing" \cite[Theorem 34]{BaganDG07}.

\begin{toappendix}
\begin{lemma}\AP\label{lem:semantic-fghw}
    A "CQ" $q(X)$ is "equivalent" to a "CQ" of "free generalized hyperwidth" $k$ if, and only if, $\core(q(X))$ is of "free generalized hyperwidth" at most $k$.
\end{lemma}
\begin{proof}
    Let $q'(X)$ be "equivalent" to $q(X)$, meaning that there exists a "$C$-homomorphisms" $h:q' \Chomto \core(q(X))$ which is ``strong onto'' in the sense that $h(q') = \atoms(\core(q(X)))$.
    Let $(T,\bagmap,\atommap)$ be a "free-connex" "tree decomposition" of $q'(X)$, and apply the mapping $h$ to both $\bagmap$ and $\atommap$ obtaining $(T,\bagmap',\atommap')$. It follows that $(T,\bagmap',\atommap')$ is also a "free-connex" "tree decomposition" of $\core(q(X))$, and the "width" of $(T,\bagmap',\atommap')$ can only be smaller or equal to that of $(T,\bagmap,\atommap)$.
\end{proof}
\end{toappendix}

\section{Aggregate Queries}
\label{sec:aggregation}

\AP
For a set $K$, a ""$K$-annotated database"" is a pair $(D,\intro*\ann)$ where $D$ is a "database" and $\ann : D \to K$ is a function giving a $K$-valued ""annotation"" to each "fact". In our examples we will fix $K = \Nat$, and hence by ""annotated database"" we shall mean a "$\Nat$-annotated database".

We now define the language of aggregate queries based on "CQs" over "$K$-annotated databases", which we call ""$K$-aggregate queries"".\footnote{This language is heavily inspired by previous works as stated in \Cref{sec:relworks}.}
A "$K$-aggregate query on a set $X$ of free variables" yields, when evaluated on a "$K$-annotated database", a set of pairs $(f,n)$ such that $f$ is an "assignment" $f: X \to \Const$ and $n \in K$. The syntax and semantics are given below:

\medskip

    \proofcase{${\Join^\otimes} q$}  Given a finite set $q$ of "atoms" and a "commutative monoid" $\otimes : K \times K \to K$, the expression
    ``${\Join^\otimes} q$''
    is a "$K$-aggregate query" on $X$, for $X =\vars(q)$.
    Given a "$K$-annotated database" $(D,\ann)$ the ""evaluation"" of ${\Join^\otimes} q$ for $q = \set{\alpha_1, \dotsc, \alpha_\ell}$ on $(D,\ann)$ is the set 
    \AP
    $\intro*\evalaggD{{\Join^\otimes} q}{D,\ann} \eqdef \set{(h,\ann(h(\alpha_1)) \otimes \dotsb \otimes \ann(h(\alpha_\ell))) \mid h : q \Chomto D}$ for $C=\const(q)$. That is, it outputs all satisfying "homomorphisms" together with the $\otimes$-multiplication of the "facts" involved.
    For example, in the "annotated database" $(D,\ann)$ of \Cref{fig:ex-intro}, $\evalaggD{{\Join^{\!\times}} \set{\schemaName{Sale}(\textit{pers},d,\textit{prd},qty), \schemaName{Price}(\textit{prd},amt)}}{D,\ann}$ contains in particular $(f,8)$, where $f = (\textit{pers},d,\textit{prd},qty,amt) \assto (\text{p1},\text{2-sept},X,4,2)$.

    \medskip
    
    \proofcase{$ \pi_{Y}^{\oplus}~\gamma $} Given a "commutative monoid" $\oplus : K \times K \to K$, a "$K$-aggregate query" $\gamma$ on $X$, and a set $Y \subsetneq X$, %
    the expression ``$ \pi_{Y}^{\oplus}~\gamma $''
    is a "$K$-aggregate query" on $Y$.
    Given a "$K$-annotated database" $(D,\ann)$ the ""evaluation"" of $\pi_{Y}^{\oplus}\gamma$ on $(D,\ann)$ is the set 
    consisting of all pairs $(f,n)$ such that $S_f = \set{(f',n') \in \evalaggD{\gamma}{D,\ann} : f'(Y) = f(Y)}$ is non-empty and $n = \bigoplus \set{n' : (f',n') \in S_f}$.%
    \footnote{
      By $f'(Y) = f(Y)$ we mean $f(y) = f'(y)$ for all $y \in Y$; and by $\bigoplus \set{n' : (f',n') \in S_f}$ we mean the $\oplus$-sum \emph{with multiplicities}: each $n' \in \set{n' : (f',n') \in S_f}$ is $\oplus$-summed $|\set{f' : (f',n') \in S_f}|$-times.
    }
    That is, it outputs the restriction onto $Y$ of all "homomorphisms" given by $\gamma$ while $\oplus$-summing the annotations of the collisions.
    Again in \Cref{fig:ex-intro}, we have that $\evalaggD{\pi_{\emptyset}^{\max}\Join^{\!\times}\set{\schemaName{Sale}(\textit{pers},d,\textit{prd},qty)}}{D,\ann}$ is $\set{(\emptyset,5)}$ where $\emptyset$ is the empty "assignment", while $\evalaggD{\pi_{\set{\textit{pers}}}^{\!+}\Join^{\!\times}\set{\schemaName{Sale}(\textit{pers},d,\textit{prd},qty)}}{D,\ann}$ is $\set{(\textit{pers} \assto \text{p1}, 10), (\textit{pers} \assto \text{p2}, 2)}$ (on atomic queries the choice for the $\times$-product is irrelevant).

\smallskip

We will simply write \reintro{aggregate query} to denote a "$K$-aggregate query" on any $K$.
\AP
The ""underlying hypergraph"" $\reintro*\hyperq[\gamma]$ of an "aggregate query"   $\gamma = \pi_{X_n}^{\oplus_n} \dotsb \pi_{X_1}^{\oplus_1} {\Join^\otimes} q$ is defined as $\hyperq$, and its size is $\intro*\sizeofgamma \eqdef \sizeofq + \sum_{i\in[n]} |X_i|$.
Further, we say that
\AP
a "$K$-aggregate query" $\gamma$ is a ""semiring $K$-aggregate query"" if $(K,\oplus_i,\otimes)$ is a "commutative semiring" for every $i$.

\AP
The ""evaluation task"" $\reintro*\evalPb{\class}$ for a class $\+C$ of "$K$-aggregate queries" is the task of, given $\gamma \in \+C$ and a "$K$-annotated database" $(D,\ann)$, computing $\evalaggD{\gamma}{D,\ann}$.
The concepts of "constant delay enumeration after a linear (or polynomial) pre-processing@constant delay enumeration after a linear pre-processing" for $\class$ are defined analogously.

\subparagraph*{Basic Examples of Aggregate Queries}
We show here some examples of what kind of queries can be 
\AP
expressed via "aggregate queries".
Other examples can be found in \Cref{app:otherexamples}.
  Note that all the examples are of "semiring aggregate queries" and hence the complexity results of this manuscript will apply to such kind of queries.
  
In what follows we consider $+$, $\max$ and $\times$ as the operations of sum, maximum, and product of natural numbers, respectively. 
\AP
We will refer to $(D,\intro*\oneann)$ as a ""$1$-initialized"" "$\Nat$-annotated database" where $\oneann$ is the function that annotates each fact in the database with 1.
\begin{toappendix}
  \subsection{Other Examples of Aggregate Queries}
  \label{app:otherexamples}
\end{toappendix}

\begin{example}[CQ]
  Any "CQ" $q(X)$ can be expressed as a "$K$-aggregate query" $\pi^{\max}_X {\Join^{\max}} q$ where $K=\set{1}$ is a singleton set of "annotations". Indeed, for any "database" $D$, we have that $\evalqD{q(X)}{D}$ coincides with the projection onto the first component of $\evalaggD{(\pi^{\max}_X {\Join^{\max}} q)}{D,\oneann}$.\footnote{Since there is only one semiring on one element, the query could be more intuitively written as $\pi_X {\Join} q$.}
\end{example}

\begin{toappendix}
  \begin{example}[Group-by counting]
    \AP
    On a $(D,\oneann)$ database, an "aggregate query" of the form $\pi^+_{\set{y}} {\Join^\times} \set{\textsf{PersonDistrict}(x,y)}$ would compute the number of inhabitants per district.
    Concretely, the evaluation of a "$\Nat$-aggregate query"  $\pi_X^+ {\Join^\times}q$ outputs pairs of the form $(f,n)$ where $f$ is an assignment for $X$ and $n$ is the number of possible assignments $g$ for the remaining variables such that $f \cup g$ is a "homomorphism". In particular, if $X = \emptyset$, then $\pi^{+}_{\emptyset}{\Join^\times}q$ counts the number of "homomorphisms" from $q$ to the "database".
  \end{example}  
\end{toappendix}

\begin{example}[Counting answers to "CQs"]
  A query of the form $\pi_\emptyset^{+} \pi_X^{\max} {\Join^\times} q$ evaluated on $(D,\oneann)$ outputs a pair $(\emptyset,n)$, where $n$ is the number of answers to the (classical) "CQ" $q(X)$. In the literature this is often written $\# q(X)$. We will study this particular case in \Cref{sec:countingcq}.
\end{example}

\begin{toappendix}  
  \begin{example}[Counting answers to quantified "CQs"]
    Another example is the case of counting solutions to a quantified conjunctive query --known as \AP""shCQC"" \cite[Example~1.3]{KhamisNR16}-- in which case as usual the existential can be modelled as `$\max$' and the universal as `$\times$', again on "$1$-initialized" "$\Nat$-annotated databases" $(D,\oneann)$.  
  \end{example}
\end{toappendix}

\begin{example}[Max group-by count]
  \label{ex:maxgroupcount}
  On $(D,\oneann)$,
  suppose we want to compute the maximum number of inhabitants satisfying some criteria in a district, then the query $\pi_{\emptyset}^{\max} \pi_{\set{y}}^{+} \pi_{\set{x,y}}^{\max} {\Join^{\!\times}}q$ with $q$ of the form
  $ 
  \set{\textsf{Person}(x,\bar w), \textsf{PersonDistrict}(x,y), \textsf{District}(y,\bar z), q_{\textup{criteria}}(\bar z,\bar w)} 
  $ %
  would yield the expected result. More generally, an "aggregate query" of the form $\pi_\emptyset^{\max} \pi_Y^{+} \pi_X^{\max} {\Join^{\!\times}}q$ outputs a pair $(\emptyset,n)$, where $n$ is the maximum number of answers to the (classical) "CQ" $q(X)$ for every possible assignment of the variables $Y \subsetneq X$. 
\end{example}

\begin{toappendix}
  \begin{example}[Witnessing provenance]
    \label{ex:witnessfacts}
    Suppose we want to evaluate a "CQ" $q(X)$ but we want not only the satisfying assignments for $X$ but also some witnessing "facts" for the satisfaction for all query atoms. For example, if we want to evaluate $q(\set x)$ for $q = \set{R(x,y),S(y,z)}$, we would like to obtain, for each assignment $\set{x \mapsto c}$ also some pair of facts $R(c,c')$ and $S(c',c'')$ from the database which justify that $c$ is part of the answer.
    While this can be done via a simple adaptation of any evaluation algorithm, for example "Yannakakis algorithm", in our setting it is just a matter of choosing the right "semiring" to compute the provenance with the right degree of detail needed. Below we show one such way of implementing this.
  
    \AP
    Let $\intro*\Mfacts$ be the set of all finite multisets of "facts".
    Let $\prec$ be a total order on the set of all "facts", and consider the following lexicographic lifting $\prec^*$ of $\prec$ to multisets of "facts": $\mset{\alpha_1 \prec \dotsb \prec \alpha_n} \prec^* \mset{\alpha'_1 \prec \dotsb \prec \alpha'_m}$ if, for some $1 \leq i \leq n$ we have $\alpha_j = \alpha'_j$ for every $j < i$ and either (i) $\alpha_i \prec \alpha'_i$ or (ii) $\alpha_i = \alpha'_i$ and $i=n < m$.
  
    Let $\max_{\prec} : \Mfacts \times \Mfacts \to \Mfacts$ be the $\prec^*$-maximum.
  
    \begin{claim}
      $(\Mfacts,\max_{\prec},\cup)$ is a "commutative semiring".
    \end{claim}

    Consider an "$\Mfacts$-annotated database" $(D,\ann)$ where $\ann(\alpha) = \mset{\alpha}$ for every $\alpha \in D$. Then the evaluation of the "aggregate query" $\pi_X^{\max_{\prec}}{\Join^\cup}q$ on $(D,\ann)$ computes all pairs of the form $(h,\rho)$ where $h$ is a solution to the "CQ" $q(X)$, and $\rho$ is some witnessing multiset of facts. In fact, not just any multiset but the $\prec^*$-maximum one.
  \end{example}

  \begin{example}
    Computing the \emph{how}-provenance, \emph{why}-provenance, or any other sort of provenance \cite{GreenT17} for a "CQ" can also be seen as an "aggregate query" in this abstract framework by instantiating the "semiring" accordingly.
  \end{example}
  
  \begin{example}[Combining]
    Suppose we want to compute a query like the one in \Cref{ex:maxgroupcount} but we want, additionally, to have the facts which maximize the query.
    Observe that 
    \begin{claim}
      $(\Mfacts, \cap, \cup)$ is a "commutative semiring"  
    \end{claim}
    If $(K_1,\oplus_1,\otimes_1)$ and $(K_2,\oplus_2,\otimes_2)$ are "commutative semirings", so is $(K_1 \times K_2,(\oplus_1 \times \oplus_2),(\otimes_1 \times \otimes_2))$, where $(c_1,c_2) (\oplus_1 \times \oplus_2) (c'_1,c'_2) = (c_1 \oplus_1 c'_1, c_2 \oplus_2 c'_2 )$ and likewise for $(\otimes_1 \times \otimes_2)$.
  
    Hence, the expression 
    \[
      \pi_{\emptyset}^{(\max,\max_{\prec})} \pi_{\set{y}}^{(+,\cap)} \pi_{\set{x,y}}^{(\max,\cap)} {\Join^{(\times,\cup)}}q
    \]
    formalizes the query, on "$(\Nat \times \Mfacts)$-annotated databases@annotated database" $(D,\ann)$ where $\ann(\alpha) = (1,\mset{\alpha})$ for every $\alpha \in D$.
  \end{example}

\end{toappendix}

\section{Project-Connex Decompositions for Aggregate Queries}
\label{sec:project-connex-intro}
\AP
We generalize the definition of "free-connex" to "aggregate queries". A "tree decomposition" $(T,\bagmap,\atommap)$ for an "aggregate query" $\gamma = \pi_{X_n}^{\oplus_n} \dotsb \pi_{X_1}^{\oplus_1} {\Join^\otimes} q$ is ""project-connex"" if
for every $i \in [n]$ there are subtrees $T_i$ such that
\begin{enumerate}[(i)]
    \item each $T_i$ is a "witness subtree" for $X_i$ and 
    \item for every $i < n$ we have $\vertex{T_{i+1}} \subseteq \vertex{T_{i}}$.
\end{enumerate}

If $X_n = \emptyset$, then $T_n$ is just the ``empty tree''.
\AP
The ""project-connex generalized hyperwidth"" of 
an "aggregate query" $\gamma$, or $\intro*\pghw(\gamma)$, 
is the minimum "width" among all its "project-connex" "tree decompositions".
Observe that a "tree decomposition" $(T,\bagmap,\atommap)$ of $\hyperq$ is "free-connex" for the "CQ" $q(X)$ (also written $\pi_X q$) if{f} it is "project-connex" for the "aggregate query" $\pi^\oplus_X {\Join^\otimes} q$ (or for $\pi^{\oplus_2}_\emptyset\pi^{\oplus_1}_X {\Join^\otimes} q$).
We will henceforth assume that whenever we are given a "project-connex" "tree decomposition" we also have available its "witness subtrees"; and conversely, when we build a "project-connex" "tree decomposition" we also produce the corresponding "witness subtrees".

We first show that, without loss of generality, decompositions can be assumed to be `small'. That is, large "project-connex" decompositions can be shrunk in an efficient way.
The following lemma is inspired by a similar useful result for "free-connex" "tree decomposition" \cite[Lemma~3.3]{berkholz2020constant}. We provide a generalization to "project-connex" and slight bound 
improvement.\footnote{The size of the "tree decomposition" is bounded by $|\vars(\gamma)|^2$ in \cite[Lemma~3.3]{berkholz2020constant}.}
\begin{lemma}\AP\label{lem:bound-decompositions-agg-queries}
    Given a "project-connex" "tree decomposition" $(T,\bagmap,\atommap)$ for an "aggregate query" $\gamma$, one can produce in linear time  another "project-connex" "tree decomposition" $(\widehat T,\widehat \bagmap,\widehat \atommap)$ for $\gamma$ of the same "width" so that $|\vertex{\widehat T}| \in \+O(\sizeofgamma)$.
\end{lemma}
\begin{proof}
    Suppose the "width" of $(T,\bagmap,\atommap)$ is $k$, and let $\gamma = \pi_{X_n}^{\oplus_n} \dotsb \pi_{X_1}^{\oplus_1} {\Join^\otimes} q$.
    Let $T'_n, \dotsc, T'_1$ be the "witness subtrees" of $T$ for the "project-connex" property, and let $g : X_1 \to \vertex{T'_1}$ be such that $x \in \bagmap(g(x))$ for every $x \in X_1$.
    Let $f : q \to \vertex{T}$ be any function such that $\vars(\alpha) \subseteq \bagmap(f(\alpha))$ for every $\alpha \in q$ ---it exists due to the "completeness condition".
    
    Let $Z$ be the union of the images of $f$ and $g$, note that $|Z| \leq |\atoms(q)| + |X_1| = \+O(\sizeofgamma)$. 
    Now we repeatedly remove parts of $T$ to obtain a `small' decomposition.
    \begin{enumerate}[(1)]
        \item If there exists some leaf "bag" $v \in \vertex{T}$ which is not in $Z$, we can simply eliminate $v$ from $T$ and from every $T'_i$ such that $v \in \vertex{T'_i}$. It follows readily that the restriction of $(T,\bagmap, \atommap)$ is still a "tree decomposition" for $\gamma$ of "width" at most $k$, and that the resulting restrictions of $T'_1, \dotsc, T'_n$ are "witness subtrees".

        \item If there is a "bag" $v \in \vertex{T}$ which is not in $Z$ and that it has only one child $v'$ in $T$, then the edge $\set{v,v'}$ can be contracted. That is, we define $(\widehat T, \widehat\bagmap, \widehat \atommap)$ where  $\widehat T$ is the result of replacing $v,v'$ with one vertex $v^*$, 
        we let
        $\widehat \bagmap \eqdef \set{u \mapsto \bagmap(u) : u \in \vertex{T} \setminus \set{v,v'}} \cup \set{v^* \mapsto \bagmap(v')}$, and similarly for $\widehat \atommap$.
        The "witness subtrees" $\widehat T'_i$ are defined to contain $v^*$ "iff" $T'_i$ contained $v'$.
        Note that all properties of "tree decomposition" are inherited from $(T,\bagmap, \atommap)$, and each $T'_i$ is still connected and covers the variables $X_i$ ("ie", it is a "witness subtree").
    \end{enumerate}
    It follows that by repeatedly performing these two simplifications in any order until saturation, we arrive at some $(\widehat T,\widehat \bagmap,\widehat \atommap)$ which is a "project-connex" "tree decomposition" for $\gamma$ of "width" ${\leq}k$ such that (i) $\set{\widehat T'_i}_i$ are the "witness subtrees" (ii) $\widehat T$ has a linear number of leaves and (iii) there is a linear number of vertices of $\widehat T$ having only one child.
    Hence, $|\vertex{\widehat T}| \in \+O(\sizeofgamma)$.

    Further, this construction can be done in linear time. We proceed bottom-up in $T$, we first mark each leaf as ``processed and deleted'' or ``processed and not deleted'' depending on whether the leaf belongs to $Z$. We then iteratively treat nodes such that all its children have been processed: if it has at least two children processed and not deleted we mark the node as processed and not deleted, otherwise we check if we are in one of the two cases above, and mark the node appropriately. Once the tree is completely marked, it is easy to produce the resulting "tree decomposition" $(T,\bagmap, \atommap)$ and the "witness subtrees" in linear time.
\end{proof}

\begin{toappendix}
The following can be seen as a corollary of the proof of \Cref{lem:bound-decompositions-agg-queries}.
\begin{lemma}
    \Cref{lem:bound-decompositions-agg-queries} can be improved to obtain, in polynomial time, a linear-size decomposition for $\gamma= \pi^{\oplus_1}_{X_1} \dotsb \pi^{\oplus_n}_{X_n} {\Join^\otimes} q$ where the size\footnote{We assume that $X_1 \neq \emptyset$, otherwise the bound is $|\vertex{T_i}| \leq \Pi_{i=2}^n  (|X_i| - |X_{i-1}|)$.} of the "witness subtrees" $(T_i)_i$ are such that $|\vertex{T_i}| \leq |X_1| \cdot \Pi_{i=2}^n  (|X_i| - |X_{i-1}|)$.
\end{lemma}
\begin{proof}
    It is known \cite{GottlobLS02} (see also \cite[proof of Lemma~3.3]{berkholz2020constant}) that for every "tree decomposition" of a "CQ" $q(X)$ one can find, in polynomial time, another "tree decomposition" of whose number of vertices is bounded by $|\vars(q)|$. Applying this construction to the  "witness subtree" $T_1$ for $X_1$ from the "project-connex" "tree decomposition" given by \Cref{lem:bound-decompositions-agg-queries}, we can bound $|\vertex{T_1}| \leq |X_1|$. In general, we can also bound the size of  each connected component of the forest resulting from deleting $\vertex{T_{i-1}}$ from the "witness subtree" $T_i$ for $X_i$. Since there are at most $|\vertex{T_{i-1}}|$ connected components, and since each one can be bounded by $|X_{i} \setminus X_{i-1}|$ by the same argument above, we obtain that $\vertex{T_i}$ is bounded by 
    \begin{align*}
        |\vertex{T_{i-1}}| \cdot |X_{i} \setminus X_{i-1}| =
        |X_1| \cdot (|X_2| - |X_1|) \dotsb (|X_{i}| - |X_{i-1}|).
    \end{align*}
\end{proof}

\begin{corollary}[of~\Cref{lem:bound-decompositions-agg-queries}]\AP\label{lem:bound-decompositions}
    Given a "free-connex" "tree decomposition" $(T,\bagmap,\atommap)$ for a "CQ" $q(X)$, one can produce in linear time\footnote{This improves on the previously known quadratic bound \cite[Lemma~3.3]{berkholz2020constant}.}  another "free-connex" "tree decomposition" $(\widehat T,\widehat \bagmap,\widehat \atommap)$ of the same "width" so that $|\vertex{\widehat T}| \in \+O(\sizeofq)$.
\end{corollary}
\end{toappendix}

\paragraph*{Computing Tree Decompositions}
\label{para:comput-tree-dec}

We now focus on computing "project-connex" "tree decompositions". We will show that this can be achieved by computing (classical) "tree decompositions" on a different query.
A \AP""frontier-path"" 
\footnote{Generalizes a notion from \cite{DBLP:conf/icdt/ChenM15}; more details in \Cref{sec:chenmengel} and \Cref{ft:frontier-path}.} 
in an "aggregate query" $\gamma = \pi_{X_1}^{\oplus_1} \dotsb \pi_{X_n}^{\oplus_n} {\Join^\otimes} q$  between two distinct "variables" $x_1$ and $x_m$ is an alternating sequence of "variables" and "atoms" of $q$
$$
    x_1 \xrightarrow{\alpha_1} x_2 \xrightarrow{\alpha_2} \dotsb \xrightarrow{\alpha_{m-2}} x_{m-1} \xrightarrow{\alpha_{m-1}} x_m
$$
such that (i) $\alpha_i \in \atoms(q)$ and $x_i \in \vars(\alpha_i)$ for every $i \in [1,m-1]$,  (ii) $x_i \in \vars(\alpha_{i-1})$ for every $i \in [2,m]$, and (iii) for some $i \in [1,n]$ we have 
$X_i \cap \set{x_1, \dotsc, x_m} = \set{x_1, x_m}$. Whenever we want to make explicit the witnessing set $X_i$ we shall write 
`\reintro{$X_i$-frontier-path}'.

\AP
We define the ""augmented query"" of $\gamma$, noted $\intro*\aug{\gamma}$, as the "Boolean" "CQ" $q()$, 
where $q$ consists of all the "atoms" of $\gamma$ plus an "atom" $\PRel(x,y)$ for every pair of variables $x,y \in \vars(\gamma)$ connected via a "frontier-path", where $\intro*\PRel$ is any fixed binary "relation name". Observe that $\aug{\gamma}$ can be computed from $\gamma$ in polynomial time, and it may be of quadratic size with respect to the size of $\gamma$.
\begin{toappendix}
    \begin{lemma}
        The "augmented query" $\aug \gamma$ can be computed from $\gamma$ in polynomial time.
    \end{lemma}
    \begin{proof}
        For every $x,y \in \vars(\gamma)$ and projection $\pi_X^\oplus$ of $\gamma$, we can test whether there is an "$X$-frontier-path" between $x$ and $y$ in linear time (it can be seen as an instance of the ``st-connectivity'' problem).
    \end{proof}
\end{toappendix}
The interest of $\aug\gamma$ is that it allows the computation of the "project-connex" "tree decomposition" of $\gamma$ via the classical "tree decomposition" of $\aug\gamma$, as shown next.\footnote{\Cref{thm:width-augmented-char} can be seen as an analog to the claim of \cite[Theorem 34]{JoglekarPR16} for a related measure. Unfortunately, the proof has a flaw and in fact the theorem is false, see \Cref{sec:problemAJAR} for details.}
\begin{theoremrep}\label{thm:width-augmented-char}
    For every "aggregate query" $\gamma$ we have $\pghw(\gamma) = \ghw(\aug\gamma)$. 
    Further, decompositions can be obtained from one another in linear time.
\end{theoremrep}
\begin{proofsketch}
    On the one hand, a "project-connex" "tree decomposition" of "width" $k$ for an "aggregate query" $\gamma$ is in fact also a "tree decomposition" for $\aug\gamma$ since it can be shown that every pair of variables connected via a "frontier-path" must be contained in a "bag". 

    On the other hand, given a "tree decomposition" $(T,\bagmap,\atommap)$ of $\aug\gamma$, one can build a "project-connex" "tree decomposition"  of $\gamma$ preserving its "width". We proceed by induction on the number of projections of $\gamma$. The base case with no projections is trivial. For the inductive case, suppose we could already build a "project-connex" "tree decomposition" $(\widehat T, \widehat\bagmap, \widehat\atommap)$ of $\pi_{X_2}^{\oplus_2} \dotsb \pi_{X_n}^{\oplus_n} {\Join^\otimes} q$ and we want to add a $\pi^{\oplus_1}_{X_1}$-projection.

    Let $C_1, \dotsc, C_k \subseteq (\vars(q) \setminus X_1)$ be the partition of $\vars(q) \setminus X_1$ induced by the  
    connected components of $\restrictG[\hyperq]{\vars(q) \setminus X_1}$.
    Let $C^+_i$ be the set of variables of $X_1$ incident to $C_i$, and note that any pair of distinct variables $x,y \in C^+_i$ is connected via a "frontier-path". This in turn implies that there are "bags" $v_1, \dotsc, v_k \in \vertex{T}$ such that $C^+_i \subseteq (\bagmap(v_i)\cap X_1)$ for every $i$.

    We can now create a "project-connex" "tree decomposition" by: (1) `""relativizing""' $(T,\bagmap,\atommap)$ onto $X_1$ ("ie", replacing $\bagmap(v)$ with $\bagmap(v) \cap X_1$ for every $v$), and (2) for each $i \in [k]$ attaching, as a child of $v_i$, a fresh copy of $(\widehat T ,\widehat \bagmap, \widehat \atommap)$ "relativized" onto $C_i \cup C_i^+$.
\end{proofsketch}
\begin{proof}
    Let $\gamma=\pi_{X_1}^{\oplus_1} \dotsb \pi_{X_n}^{\oplus_n} {\Join^\otimes} q$.

    \proofcase{$\pghw(\gamma) \geq \ghw(\aug\gamma)$}
    We show how to build a "tree decomposition" of $\aug\gamma$ from a "project-connex" "tree decomposition" of $\gamma$ preserving its "width".
    Let $(T,\bagmap,\atommap)$ be a "width" $k$ "project-connex" "tree decomposition" for $\gamma$ and let $T_i$ be the "witness subtree for" each $X_i$. We claim that $(T,\bagmap,\atommap)$ is in fact a "tree decomposition" for $\aug\gamma$. It suffices to show that for every pair $x,y$ of variables connected via a "frontier-path" there is a "bag" $v \in \vertex{T}$ so that $\set{x,y} \subseteq \bagmap(v)$. 
    Suppose $x,y \in X_i$ and the witnessing "$X_i$-frontier-path" is
    $$
        x_1 \xrightarrow{\alpha_1} x_2 \xrightarrow{\alpha_2} \dotsb \xrightarrow{\alpha_{m-2}} x_{m-1} \xrightarrow{\alpha_{m-1}} x_m
    $$    
    If $m=2$, then $x,y$ appear in the same "atom" and hence by the "completeness condition" on $(T,\bagmap,\atommap)$ there must be a "bag" containing them.
    
    Suppose then that $m>2$.
    The "frontier-path" induces a ``path'' $(v_1,x_1), \dotsc, (v_m,x_m)$ in $T$, where (i) $x_i \in \bagmap(v_i)$ for every $i$, and (ii) for every $i \in [1,m-1]$, either (a) $v_i = v_{i+1}$ or (b) $x_i = x_{i+1}$ and $\set{v_i,v_{i+1}} \in \edges{T}$. 
    Let $i_x$ be the maximum index such that $x_{i_x} = x$ and $v_{i_x} \in \vertex{T_i}$ and $i_y$ be the minimum index such that $x_{i_y} = y$ and $v_{i_x} \in \vertex{T_i}$. Note that $i_x < i_y$. By definition of "frontier-path", $v_{i_x}, \dotsc, v_{i_y}$ is a path in $T$ whose only vertices in $T_i$ are the first and last (and since $T_i$ is a connected subtree of $T$), we must have that $v_{i_x} = v_{i_y}$ ---otherwise $T$ would not be a tree.
    Hence, $\set{x,y} \subseteq \bagmap(v_{i_x})$.

    \bigskip

    \proofcase{$\pghw(\gamma) \leq \ghw(\aug\gamma)$}
    We show how to build a "project-connex" "tree decomposition" of $\gamma$ from a "tree decomposition" $(T,\bagmap,\atommap)$ of $\aug\gamma$ preserving its "width". 
    Further, we show that for every set $C \subseteq \vars(\gamma)$ which forms a ""$\PRel$-clique"" in $\aug\gamma$ ("ie", $\PRel(x,y)\atoms(\aug\gamma)$ for every pair of distinct $x,y \in C$), the produced decomposition has a "bag" containing $C$ ---we shall refer to this as the \AP""$\PRel$-clique property"".
    We proceed by induction on the number $n$ of projections. 
    
    The case $n=0$ where there are no projections is trivial, since in this case a "project-connex" "tree decomposition" is just a "tree decomposition" and the "augmented query" does not add any "atom". Observe that the "$\PRel$-clique property" is met.
    
    For the inductive case, suppose we could already build a "project-connex" "tree decomposition" $(\widehat T, \widehat\bagmap, \widehat\atommap)$ of $\pi_{X_2}^{\oplus_2} \dotsb \pi_{X_n}^{\oplus_n} {\Join^\otimes} q$, with $\set{\widehat W^{(j)}}_{j \in [2,n]}$ as "witness subtrees". Assume also that the "$\PRel$-clique property" holds.

    Let $\bagmap'$ be the \AP""relativization"" of $\bagmap$ onto $X_1$ ("ie", $\bagmap'(v) \eqdef \bagmap(v) \cap X_1$ for all $v \in \vertex{T}$), and observe that $(T,\bagmap',\atommap)$ verifies the "connectivity@connectivity condition" and "covering@covering condition" conditions of "tree decompositions".

    Let $C_1, \dotsc, C_k \subseteq (\vars(q) \setminus X_1)$ be the partition of $\vars(q) \setminus X_1$ induced by the  connected components of $\restrictG[\hyperq]{\vars(q) \setminus X_1}$.
    For each $i$, let $C^+_i$ be the set of variables of $X_1$ incident to $C_i$, that is, $C^+_i \eqdef \set{x \in X_1 : \set{x,x'} \subseteq e \text{ for some } x' \in C_i  \text{ and } e \in \edges{\hyperq}}$.
    Observe that for each $i$, every pair of distinct variables $x,y \in C^+_i$ is connected via a "frontier-path", and thus $\PRel(x,y) \in \aug\gamma$. In other words, $\aug\gamma$ contains a "$\PRel$-clique" on each $C^+_i$. This means that there must be "bags" $v_1, \dotsc, v_k \in \vertex{T}$ such that $C^+_i \subseteq \bagmap'(v_i)$ for every $i$ due to \Cref{lem:cliques-in-bags}.

    For $i \in [1,k]$, let $(\widehat T_i,\widehat \bagmap_i,\widehat \atommap_i)$ be pairwise vertex-disjoint copies of $(\widehat T ,\widehat \bagmap, \widehat \atommap)$ (we assume also vertex-disjointness with $T$), and let $\widehat \bagmap'_i$ be the "relativization" of $\widehat \bagmap_i$ onto $C_i \cup C_i^+$ ("ie", $\widehat \bagmap'_i(v) \eqdef \widehat \bagmap_i(v) \cap (C_i \cup C_i^+)$ for all $v \in \vertex{\widehat T_i}$). 
    By the "$\PRel$-clique property", for every $i$ there is a "bag" $\widehat v_i \in \vertex{\widehat T_i}$ such that $C^+_i \subseteq \widehat\bagmap'_i(\widehat v_i)$.

    We now produce a ``collage of trees'' $\widetilde T$ as the result of taking the disjoint union of $T,\widehat T_1, \dotsc, \widehat T_k$, and adding an edge between $v_i$ and $\widehat v_i$ for every $i \in [1,k]$. We finally produce $(\widetilde T, \widetilde\bagmap, \widetilde\atommap )$ 
    for $\widetilde \bagmap = \bagmap' \cup \bigcup_{i=1}^k \widehat\bagmap'_i$ and $\widetilde \atommap = \atommap \cup \bigcup_{i=1}^k \widehat \atommap_i$. 
    \begin{claim}
        $(\widetilde T, \widetilde\bagmap, \widetilde\atommap )$ is a "project-connex" "tree decomposition" of $\gamma$. 
    \end{claim}
    \begin{nestedproof}
        \proofcase{Tree property.} Indeed, $\widetilde T$ is a tree because each $v_i$ is attached to a different tree $\widehat T_i$, hence no cycles can appear.

        \proofcase{"Completeness condition".} Let $e \in \edges{\hyperq}$. If $e \subseteq X_1$, then there must be some vertex $v$ of $T$ such that $\bagmap'(v) \supseteq e$. Otherwise, $e \subseteq  (C_i \cup C_i^+)$ for some $i$ and by the "completeness condition" of $(\widehat T, \widehat \bagmap, \widehat \atommap)$ there must be some $v \in \vertex{\widehat T}$ such that $\widehat\bagmap(v) \supseteq e$. Hence, the corresponding copy of $v$ in $\widehat T_i$ must be such that $\widetilde \bagmap(v) = \widehat \bagmap'_i(v) = \widehat\bagmap(v) \cap (C_i \cup C_i^+) \supseteq e$.

        \proofcase{"Connectivity condition".} Take any variable $x \in \vars(\gamma)$. If $x \in X_1$, then $x$ is "connected@@tdec" in $(T,\bagmap',\atommap)$. 
        If $x$ appears in some $\widehat T_i$, then must $x$ belong to $C_i^+$ by definition of $\widehat T_i$. In this case, the "connectivity@@tdec" of $x$ is inherited from "connectivity@@tdec" of $x$ on $(\widehat T_i,\widehat \bagmap'_i, \widehat\atommap_i)$ and the fact that the edge $\set{\widehat v_i, v_i}$ between $\widehat T_i$ and $T$ links two "bags" "containing@@bag" $x$.
        If, on the other hand, $x \not\in X_1$, then $x$ belongs to some $C_i$, and thus "connectivity@@tdec" is inherited from $(\widehat T_i,\widehat \bagmap'_i, \widehat\atommap_i)$.

        \proofcase{"Covering condition".} Each vertex $v \in \vertex{\widetilde T}$ is such that $\widetilde \bagmap(v)$ and $\widetilde \atommap(v)$ come from a "tree decomposition" verifying the "covering condition" with the exception that the bag has been perhaps shrunk.%
        Hence, the "covering condition" is inherited from the corresponding "tree decomposition".

        \proofcase{"Project-connectedness".}
        Let $\widehat W^{(j)}_i$ be the copy of $\widehat W^{(j)}$ in $\widehat T_i$.
        We define the "witness subtree" $\widetilde W^{(1)}$ for $X_1$ to be simply $T$.
        For every $j \in [2,n]$, we define the "witness subtree" $\widetilde W^{(j)}$ as the subtree of $\widetilde T$ induced by the union $\vertex{T} \cup \bigcup_{i=1}^k \vertex{\widehat W^{(j)}_i}$. It is easy to see that 
        \begin{enumerate}
            \item $\bigcup_{v\in \vertex{\widetilde W^{(j)}}}\widetilde\bagmap(v) = X_j$ for every $j \in [1,n]$, 
            \item the nesting $\vertex{\widetilde W^{(1)}} \subseteq \dotsb \subseteq \vertex{\widetilde W^{(n)}}$ is verified, and 
            \item each $\widetilde W^{(j)}$ is a connected subtree of $\widetilde T$.
        \end{enumerate}
    \end{nestedproof}
    \begin{claim}\label{cl:characterization-augmented-preserves-width}
        The "width" of $(\widetilde T, \widetilde\bagmap, \widetilde\atommap )$ equals the "width" of $(T,\bagmap,\atommap)$.
    \end{claim}
    \begin{nestedproof}
        \proofcase{($\leq$)} This follows from the fact that for every "bag" $v$ of $\widetilde T$ we have $\widetilde\atommap(v)$ is equal to either $\atommap(v')$ or $\widehat\atommap(v')$ for some $v'$.

        \proofcase{($\geq$)} On the other hand, for every "bag" $v$ of $T$, we have that $v$ is also in $\widetilde T$ and further $\bagmap(v) = \widetilde\bagmap(v)$.
    \end{nestedproof}
    \begin{claim}
        $(\widetilde T, \widetilde\bagmap, \widetilde\atommap )$ has the "$\PRel$-clique property".
    \end{claim}
    \begin{nestedproof}
        Observe that any maximal "$\PRel$-clique" $Z$ on $\aug\gamma$ is, necessarily, equal to $C_i^+$ for some $i \in [1,k]$. Since $\widetilde \bagmap(v_i)$ contains $C_i^+$, the claim follows.
    \end{nestedproof}

    \begin{claim}
        The "project-connex" "tree decomposition" for $\gamma$ can be computed in linear time.
    \end{claim}
    \begin{nestedproof}
        Given the "tree decomposition" $(T,\bagmap,\atommap)$ of $\aug\gamma$, observe that we build $(\widehat T, \widehat \bagmap, \widehat \atommap)$ by adding edges between tree decompositions of the form ``$(T,\bagmap \cap Z,\atommap)$'', where $\bagmap \cap Z$ is the "relativization" of $\bagmap$ onto $Z$, and $Z$ can either be (i) some $X_i$ or (ii) some $C_i \cup C_i^+$. It is not hard to see that there are only a linear number of such $Z$.
        Hence, there is a linear number of $(T,\bagmap \cap Z,\atommap)$ decompositions, each $(T,\bagmap \cap Z,\atommap)$ can be obtained from $(T,\bagmap,\atommap)$ in linear time, and the set of all $Z$'s can be also obtained in linear time. After adding a (linearly many) edges we obtain $(\widehat T, \widehat \bagmap, \widehat \atommap)$, so the total cost of producing it remains in linear time.
    \end{nestedproof}
This completes the proof of \Cref{thm:width-augmented-char}.
\end{proof}

\paragraph*{Computing Acyclic Decompositions}
\label{para:comp-acyclic-dec}
In the context of "conjunctive queries", a notable characterization by Brault-Baron shows: 
\begin{theorem}[{\cite[Theorem 13, first two items]{brault2013pertinence}, see also \cite[Theorem 5.2]{berkholz2020constant}}]
    \label{prop:brault-baron-char}
    A "CQ" $q(X)$ has a "width" 1 "free-connex" "tree decomposition" if{f} both $q$ and $q \cup \set{R(\bar x)}$ have "join trees", where $R$ is a fresh "relation name" and $\vars(\bar x) = X$.    
\end{theorem}
 In other words, an "aggregate query" of the form $\pi_X^\oplus {\Join^\otimes} q$ has a "width" 1 "project-connex" "tree decomposition" if{f} both $q$ and $q \cup \set{R(\bar x)}$ have "join trees", where $R$ is a fresh "relation name" and $\vars(\bar x) = X$.
A corollary of \Cref{thm:width-augmented-char} shows that in fact, it is not necessary to test for \emph{two} different queries for "free-connex" "acyclicity" but just for one:
\begin{corollary}
    A "CQ" $q(X)$ is "free-connex" "acyclic" if{f} $\aug \gamma$ is "acyclic", for $\gamma = \pi_X^\oplus {\Join^\otimes} q$.
\end{corollary}

However, the characterization in \Cref{prop:brault-baron-char} leads to a linear-time algorithm for testing "free-connex" "acyclicity", since 
testing if a "CQ" is "acyclic" is in linear time via ``GYO decompositions'' \cite{abiteboul1995foundations} (producing a "width" 1 "tree decomposition" ---"aka" ``"join tree"''--- if successful).
In our case, building $\aug\gamma$ seems to be at least cubic-time, and 
its size may be quadratic,
even when restricted to "conjunctive queries". Hence, \Cref{thm:width-augmented-char} only leads to a polynomial algorithm for testing (and producing) "acyclic" "free-connex" "tree decompositions". 
The noticeable difference comes from the fact that $\aug\gamma$ adds cliques on a binary relation, whereas \Cref{prop:brault-baron-char} adds a single atom, hence increasing the size of the query only linearly.
We will show next that for the specific case of "acyclicity" ("ie", "project-connex" "generalized hyperwidth" 1) a modest generalization of \Cref{prop:brault-baron-char} can be shown to hold, which will enable a linear-time algorithm for finding a decomposition (provided the number of projection operators is bounded).

\begin{lemmarep}\AP\label{lem:char-project-connex-acyclic}
    An "aggregate query" $\gamma=\pi_{X_1}^{\oplus_1} \dotsb \pi_{X_n}^{\oplus_n} {\Join^\otimes} q$ has a "width" 1 "project-connex" "tree decomposition" if{f} every 
    "CQ" $\set{q_i}_{0 \leq i \leq n}$  has a "join tree", where 
    $q_0 = q$ and $q_i=q \cup \set{R_i(\bar x_i)}$ for every $i\in[1,n]$, where each $R_i$ is a fixed "relation name" and $\vars(\bar x_i) = X_i$.

    Further, the "width" 1 "project-connex" "tree decomposition" and its "witness subtrees" can be produced from the "join trees" in linear time.
\end{lemmarep}
\begin{proofsketch}
    Given a "width" 1 "project-connex" "tree decomposition" $(T,\bagmap,\atommap)$ for $\gamma$ with $\set{T_i}_i$ as "witness subtrees" and $m \leq n$, consider $(\widehat T,\widehat \bagmap,\widehat \atommap)$ to be the result of: (1)
        removing all edges between elements of $\vertex{T_m}$, 
        (2)
        adding a new vertex $v_m$ with $\widehat\bagmap(v_m) \eqdef X_m$, $\widehat\atommap(v_m) \eqdef \set{X_m}$, 
        and (3)
        adding an edge $\set{v_m,v}$ for every $v \in \vertex{T_m}$.
    It follows that $(\widehat T,\widehat \bagmap,\widehat \atommap)$ is a "tree decomposition" of "width" 1 for $q_m$, and hence that $q_m$ has a "join tree".

    On the other hand, if every "CQ" $q_i$ has a "join tree" we can produce a "width" 1 "project-connex" "tree decomposition" by follwowing the same strategy as \cite[proof of Theorem~5.2]{berkholz2020constant}.
    By induction on $n$, suppose that we have a "project-connex" "tree decomposition" $(T,\bagmap, \atommap)$ for $\gamma'=\pi_{X_{1}}^{\oplus_{1}} \dotsb \pi_{X_{n-1}}^{\oplus_{n-1}} {\Join^\otimes} q$ and a "join tree" $(T',\bagmap', \atommap')$ for $q_n$ with the corresponding bijection $\phi : \vertex{T} \to q_n$. 
    Let $f$ be any mapping $f: q \to \vertex{T}$ such that $\vars(\alpha) \subseteq \bagmap(f(\alpha))$ for every $\alpha \in q$ ---it exists due to the "completeness condition". Let $\tilde v \in \vertex{T'}$ be such that $\phi(\tilde v) = R_n(\bar x_n)$. We can produce a "project-connex" "tree decomposition" for $\gamma=\pi_{X_{1}}^{\oplus_{1}} \dotsb \pi_{X_n}^{\oplus_n} {\Join^\otimes} q$ by (i) removing $\tilde v$ from $T'$, (ii) attaching each neighbor $v$ of $\tilde v$ to $f(v)$ from $T$ and (iii) "relativizing" $T$ to $X_n$, as shown in \Cref{fig:charac-pc-width1body}.
    \begin{figure}%
        \includegraphics[width=1\textwidth]{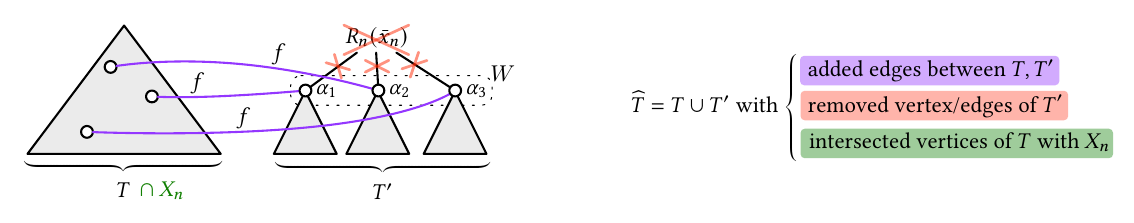}
        \caption{Construction of $(\widehat T,\widehat \bagmap,\widehat \atommap)$ in \Cref{lem:char-project-connex-acyclic}.}
        \label{fig:charac-pc-width1body}
    \end{figure}
\end{proofsketch}

\begin{proof}
    \proofcase{Left-to-right}
    Let $(T,\bagmap,\atommap)$ be a "width" 1 "project-connex" "tree decomposition" for $\gamma$ and let $T_i$ be the "witness subtree for" each $X_i$. We show that $q_m$ has a "join tree", for every $m \leq n$.
    Consider $(\widehat T,\widehat \bagmap,\widehat \atommap)$ to be the result of
    \begin{itemize}
        \item removing all edges between elements of $\vertex{T_m}$,
        \item adding a new vertex $v_m$ with $\widehat\bagmap(v_m) \eqdef X_m$, $\widehat\atommap(v_m) \eqdef \set{X_m}$,
        \item adding an edge $\set{v_m,v}$ for every $v \in \vertex{T_m}$.
    \end{itemize}

    It follows that $(\widehat T,\widehat \bagmap,\widehat \atommap)$ is a "tree decomposition" of "width" 1 for $q_m$, and hence that $q_m$ has a "join tree" by \cref{lem:jointree:width1=jointree} of \Cref{lem:jointree}.

    \medskip

    \proofcase{Right-to-left}
    We proceed by induction on the number $n$ of projections. 
    The base case of $n=0$ is trivial from the fact that in this case a "width" 1 "project-connex" "tree decomposition" is just a "width" 1 "tree decomposition", and that a "join tree" for ${\Join^\otimes} q$ is in particular a "width" 1 "tree decomposition".
    We focus then on the inductive case.
    Suppose then that we have a "project-connex" "tree decomposition" $(T,\bagmap, \atommap)$ for $\gamma'=\pi_{X_{1}}^{\oplus_{1}} \dotsb \pi_{X_{n-1}}^{\oplus_{n-1}} {\Join^\otimes} q$ and a "join tree" $(T',\bagmap', \atommap')$ for $q_n$ with the corresponding bijection $\phi : \vertex{T} \to q_n$. 
    Let $T_1, \dotsc, T_{n-1}$ be the "witness subtrees for" $X_1, \dotsc, X_{n-1}$, respectively.
    
    We need to produce a 
    "project-connex" "tree decomposition" for $\gamma=\pi_{X_{1}}^{\oplus_{1}} \dotsb \pi_{X_n}^{\oplus_n} {\Join^\otimes} q$. We assume $\vertex{T} \cap \vertex{T'} = \emptyset$, and we follow the same strategy as \cite[proof of Theorem~5.2]{berkholz2020constant}.

    Let $f$ be any mapping $f: q \to \vertex{T}$ such that $\vars(\alpha) \subseteq \bagmap(f(\alpha))$ for every $\alpha \in q$ ---it exists due to the "completeness condition".

    Let $\tilde v \in \vertex{T'}$ be such that $\phi(\tilde v) = R_n(\bar x_n)$, and let $W$ be the set of all $\tilde v$'s neighbors in $T'$. 

    We can now define the "tree decomposition" $(\widehat T,\widehat \bagmap,\widehat \atommap)$:
    \begin{itemize}
        \item the tree $\widehat T$ consists of
        \begin{itemize}
            \item all the vertices and edges from $T$,
            \item all the vertices from $\vertex{T'} \setminus \set v$, and the $T'$-induced edges,
            \item the edge $\set{v,f(v)}$ for every $v \in W$.
        \end{itemize}
        \item We define $\widehat \bagmap$ as follows
        \begin{itemize}
            \item $\widehat \bagmap(v) \eqdef \bagmap(v) \cap X_n$ for every $v \in \vertex{T}$
            \item $\widehat \bagmap(v) \eqdef \bagmap'(v)$ for every $v \in \vertex{T'}$.
        \end{itemize} 
        \item We define $\atommap$ as follows
        \begin{itemize}
            \item $\widehat \atommap(v) \eqdef \atommap(v)$ for every $v \in \vertex{T}$
            \item $\widehat \atommap(v) \eqdef \atommap'(v)$ for every $v \in \vertex{T'}$.
        \end{itemize}
    \end{itemize}

    Finally, we define the "witness subtree for" $X_n$ to be the tree $T_n$ induced by $\vertex{T}$. See \Cref{fig:charac-pc-width1} for an explanatory illustration of the construction.
    \begin{figure}
        \includegraphics[width=\textwidth]{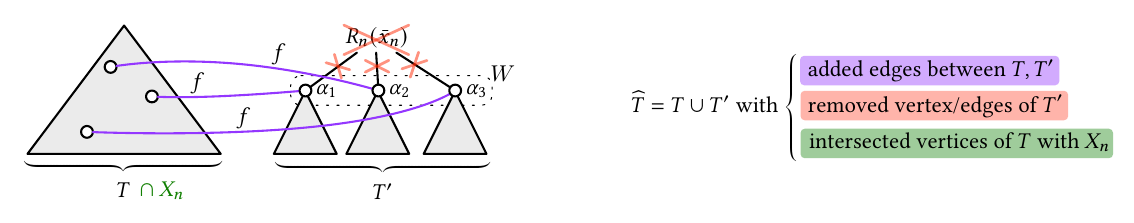}
        \caption{Construction of $(\widehat T,\widehat \bagmap,\widehat \atommap)$ in \Cref{lem:char-project-connex-acyclic}.}
        \label{fig:charac-pc-width1}
    \end{figure}

    \begin{claimrep}
        $(\widehat T,\widehat \bagmap,\widehat \atommap)$ is a "tree decomposition" of $\gamma$.
    \end{claimrep}
    \begin{nestedproof}[Proof of claim]
        The fact that $\widehat T$ is a tree follows from the fact that no two distinct vertices of $W$ belong to the same connected component of $T' \setminus \tilde v$ ("ie", the result of removing $\tilde v$ from $T'$).

        The "completeness condition" is inherited from $(T',\bagmap',\atommap')$: for each atom $\alpha$ we have $\vars(\alpha) \subseteq \widehat\bagmap(\phi^{-1}(\alpha)) = \bagmap'(\phi^{-1}(\alpha))$. 
        
        The "connectivity condition" for any variable $x \in X_n$ is the result of observing that (i) relativizing the bags of $(T,\bagmap,\atommap)$ to $X_n$ preserves the "connectivity condition" for $x$, and (ii) for any new edge $\set{v,f(v)}$, $x \in \widehat\bagmap(v)$ if{f} $x \in \widehat\bagmap(f(v))$; hence the property follows by the inherited "connectivity condition" for $x$ of $(T',\bagmap',\atommap')$.
        
        The "connectivity condition" for any variable $x' \not\in X_n$ comes from noticing (i) $x \not\in \widehat\bagmap(v)$ for every $x' \in \vertex{T}$ and (ii) $(T',\bagmap',\atommap')$ has the "connectivity condition" for $x'$. This means that for every distinct pair $v,v' \in W$ we cannot have $x' \in \bagmap'(v) \cap \bagmap'(v')$ (otherwise this would imply that $x' \in X_n$). Hence, the "connectivity condition" for $x'$ is inherited from the "connectivity condition" of the different connected components of $(T',\bagmap',\atommap')$ after removing $\tilde v$.
        
        Finally, the "covering condition" is straightforward by noticing that removing elements from the "bag" preserves the "covering condition".    
    \end{nestedproof}

    \begin{claimrep}
        $(\widehat T,\widehat \bagmap,\widehat \atommap)$ is "project-connex".
    \end{claimrep}
    \begin{nestedproof}[Proof of claim]
        Observe that $T_1, \dotsc, T_{n-1}$ are still "witness subtrees for" each $\set{X_i}_{i \leq n-1}$: indeed when relativizing $(T,\bagmap,\atommap)$ onto $X_n$ the bags for these subtrees remain unaltered since $X_n \supseteq X_i$ for every $i$.
        On the other hand $T_n$ is a "witness subtree for" $X_n$ since $\bigcup_{v \in \vertex{T}} \widehat\bagmap(v) = X_n$. Further, $\vertex{T_n} \supseteq \vertex{T_i}$ for every $i$, and thus $T_1, \dotsc, T_n$ witness the fact that $(\widehat T,\widehat \bagmap,\widehat \atommap)$ is "project-connex".
    \end{nestedproof}
    This concludes the proof.
\end{proof}

\begin{corollaryrep}
    \AP\label{cor:lineartime-width1-pc-decompositions}
    For every "aggregate query" $\gamma = \pi_{X_1}^{\oplus_1} \dotsb \pi_{X_n}^{\oplus_n} {\Join^\otimes} q$ we can test in $\+O(n \cdot \sizeofq)$ whether $\gamma$ has a "width" 1 "project-connex" "tree decomposition". Further, if possible, one such decomposition as well as the "witness subtrees" can be produced within the same time bounds.
\end{corollaryrep}
\begin{proof}
    This is a direct consequence of \Cref{lem:char-project-connex-acyclic} and \cref{lem:jointree:lineartime} of \Cref{lem:jointree}.
\end{proof}

\begin{toappendix}
    \subsection{Previously Claimed Width Characterization}
\label{sec:problemAJAR}
    \cite{JoglekarPR16} defines a set of "hypergraphs" of an AJAR query, called ""characteristic hypergraphs"" \cite[Definition 33]{JoglekarPR16}, and \cite[Theorem 34]{JoglekarPR16} claims that the width of the ``decomposable''\footnote{This is very close to our "project-connex" "tree decompositions", although for this counterexample it suffices to know that it is a restriction of "tree decompositions".}  tree decomposition  coincides with the  maximum width of the "characteristic hypergraphs". Further, this is claimed to hold for any monotonic width measure. In particular, for "generalized hyperwidth".
    
    \begin{claim}
        \cite[Theorem 34]{JoglekarPR16} does not hold.
    \end{claim}
    \begin{proof}
        The theorem fails in the sense that
        $$
        \max\set{ \ghw(G) : G\text{ is a "characteristic hypergraph" of }(H,\alpha) } \geq \ghw(H).
        $$
        is not necessarily true. 
    
        Concretely, take the AJAR query $(H,\alpha)$ such that $H$ is a clique on 4 vertices $x_1, x_2, x_3, x_4$ ("ie", $H = \set{ \set{x_i, x_j} : i  \neq j }$) and $\alpha = (x_4,\oplus)$. The "characteristic hypergraphs" are just two:
        \begin{align*}
             \+H_0 &= \set{ \set{x_1, x_2}, \set{x_2,x_3}, \set{x_1, x_3}, \set{x_1, x_2, x_3} } \hspace{1em} \text{and}\\
             \+H_1^+ &= \set{ \set{x_1, x_4}, \set{x_2, x_4}, \set{x_3,x_4}, \set{x_1, x_2, x_3} }.
        \end{align*}
    
    Hence, 
    \begin{align*}
        max\set{ \ghw(\+H) : \+H \text{ "characteristic hypergraph" of } G } =\\
         =\max\set{ \ghw(\+H_0), \ghw(\+H_1^+) } = \max\set{ 1, 2 } = 2
    \end{align*}
         
    but
       $\ghw(H) = 3$.
    
    Further, this difference can be made arbitrarily large (by having a clique on more variables).     
    \end{proof}

While in principle this would jeopardize the subsequent sections \S5.2, \S5.3, \S5.4 of \cite{JoglekarPR16}, we observe that due to our result of \Cref{thm:width-augmented-char}, a similar statement to that of 
    \cite[Corollary 35]{JoglekarPR16} can be established, and thus we do not see any obstacle to obtaining the results contained in  \cite[\S5.2--\S5.4]{JoglekarPR16}.

\end{toappendix}

\section{Evaluation of Aggregate Queries}
\label{sec:evalAgg}

In the previous section we have seen that "project-connex" decompositions can be computed effectively and that they can be assumed to be small. 
We now study the complexity of the "evaluation problem" for "semiring aggregate queries" by means of exploiting such "tree decomposition". To this end, we will use the standard approach of the semi-join, although it needs to be adapted to the "semiring" at hand. "Wlog" we assume there are no constants in queries (see \Cref{sec:constants} for details).

\begin{toappendix}
\begin{assumption}
    We will henceforth assume for simplicity that there are no constants in our queries. This is without loss of generality, see \Cref{sec:constants} for details.
\end{assumption}    
\end{toappendix}

\begin{toappendix}
    \subparagraph*{Semiring Semi-join}
\end{toappendix}
We shall first define the semi-join for any semiring and establish the complexity bound for computing it.
\AP
Given a "commutative semiring" $\aK = (K,\oplus,\otimes)$, a  ""$\aK$-semi-join"" is a "semiring $K$-aggregate query" of the form $\pi^{\oplus}_X {\Join^\otimes} \set{R(\bar x), S(\bar y)}$ for $X = \vars(\bar x)$, which we will denote ``$R(\bar x) \intro*\sjoin[\aK] S(\bar y)$'', where $\bar x$ and $\bar y$ may contain common variables. 
Observe that the classical Yannakakis ""semi-join"" can be seen as the "$\aK$-semi-join" on the "trivial semiring" $\aK = (\set{1},\oplus,\otimes)$. The following bound is obtained by extending the standard algorithm for semi-joins. 
\begin{toappendix}
    \AP
    For a "commutative monoid" $\oplus$ over $K$ and a mapping $g : X \to K$, we write $\reintro*\aboplus{x \in X}{g(x)}$ to denote $g(c_1) \oplus \dotsb \oplus g(c_\ell)$ where $c_1, \dotsc, c_\ell$ is the list without repetitions of all the elements of $X$ (the order is irrelevant due to commutativity). In particular, $\aboplus{x \in \emptyset}{g(x)} = 0_\oplus$, where $0_\oplus$ is the identity element of $\oplus$. We do likewise for $\intro*\abotimes{x \in X}{g(x)}$.
\end{toappendix}

\begin{lemmarep}\AP\label{lem:semijoin-complexity}
    For every "commutative semiring" $\aK = (K,\oplus,\otimes)$, the evaluation of a "$\aK$-semi-join" on a "$K$-annotated database" $(D,\ann)$ is in $\+O(\maxSize{D})$ (in combined complexity and under "RAM" model). It further admits "constant delay enumeration after a linear pre-processing".
\end{lemmarep}
\begin{proof}
    This is in $\+O(\maxSize{D} \cdot \log(\maxSize{D}))$ for Turing machines (by first appropriately sorting the relations). However, the "semi-join" is in $\+O(\maxSize{D})$ under the "RAM" regime.

    Suppose we are given $R(\bar x \bar y) \intro*\sjoin[\aK] S(\bar y \bar z)$, where $\vars(\bar x) \cap \vars(\bar z) = \emptyset$ (and note that we can assume this form without loss of generality).
    We first compute an array $A$ of the dimension of $\bar y$ such that $A[\bar t]$ contains
    \AP
    $\intro*\aboplus{S(\bar t \bar t') \in D}{\ann(S(\bar t \bar t'))}$.
    For this, we iterate over all the $S$-"facts" $S(\bar t \bar t')$, at each iteration we check if $A[\bar t]$ contains $\bot$.\footnote{This implies iterating over each constant of the fact $S(\bar t \bar t')$ and adding the corresponding index to $A$.}
     If so, then write $A[\bar t] = \ann(S(\bar t \bar t'))$; otherwise write $A[\bar t] = A[\bar t] \oplus \ann(S(\bar t \bar t'))$. Each element of $A[\bar t]$ can be accessed in time $O(\dim(\bar y))$ and thus the process is in $\+O(n_S \cdot \dim(\bar y)) \leq \+O(\maxSize{D})$, where $n_S$ is the number of $S$-"facts".

    Now we use an array $B$ of the dimension of $\bar x \bar y$ to encode the evaluation output.
    We iterate over each $R$-"fact" $R(\bar t' \, \bar t)$ of $D$ and we obtain $c = A[\bar t]$. If $c \neq \bot$, we write $B[\bar t' \bar t] = \ann(R(\bar t' \, \bar t)) \otimes c$. 
    This second process takes 
    $\+O(n_R \cdot (\underbrace{\dim(\bar y)}_{\text{access to $A[\bar t]$}} + \underbrace{\dim(\bar x \bar y)}_{\text{access to $B[\bar t' \bar t]$}})) \leq \+O(\maxSize{D})$,
    where $n_R$ is the number of $R$-"facts".

    The final output is described by $\set{((\bar x \bar y \assto \bar t' \bar t), n) : B[\bar t' \bar t] = n \neq \bot}$. 
    The output admits "constant delay enumeration after a linear pre-processing", by simply adding, as part of the pre-processing, another iteration of the $R$-"facts" $R(\bar t' \, \bar t)$ of $D$, checking if $B[\bar t' \bar t] \neq \bot$, and if so adding the pair $((\bar x \bar y \assto \bar t' \bar t), B[\bar t' \bar t])$ to a chained list $\+L$ (initially empty).
    The output can now be enumerated by writing, one by one, all the elements of $\+L$.
\end{proof}

We can now show how the algorithm works on "aggregate queries" which are "project-connex-acyclic@project-connex generalized hyperwidth" yielding bounds on \emph{combined complexity} evaluation.

\begin{lemma}\AP\label{lem:evalAggQuery:width1}
    Given a "$K$-annotated database" $(D,\ann)$, a %
    "semiring $K$-aggregate query" 
    $\gamma = \pi_{X_n}^{\oplus_n} \dotsb \pi_{X_1}^{\oplus_1} {\Join^\otimes} q$,
    and a "project-connex" "join tree" $(T,\bagmap,\atommap)$ thereof, 
    the "evaluation" result $\evalaggD{\gamma}{D,\ann}$ admits a "constant delay enumeration after a linear pre-processing".
    If $X_n=\emptyset$, the evaluation result can be obtained in $\+O\big(\sizeofgamma \cdot \maxSize{D} \big)$ in "combined complexity".
\end{lemma}
\begin{proof}
    This is essentially "Yannakakis algorithm", with the exception that now we need to choose the "semiring" $\aK$ for the "$\aK$-semi-join" at each application. %
    Let $\phi: \vertex{T} \to q$ be the "witnessing bijection" of the "join tree".
    If $X_1 = \vars(q)$, then $\gamma$ is "equivalent" to $\pi_{X_n}^{\oplus_n} \dotsb \pi_{X_2}^{\oplus_2} {\Join^\otimes} q$. So let us suppose there is a leaf $v$ of $T$ and a corresponding atom $\alpha_v = \phi(v)$ such that
    such that $\vars(\alpha_v) \not\subseteq X_1$. Let $v'$ be the parent of $v$ and $\alpha_{v'} = \phi(v')$ the corresponding "atom" of $q$. We now perform a "$\aK$-semi-join" for $\aK = (K,\oplus_{1},\otimes)$ on $\alpha_{v'} \sjoin[\aK] \alpha_v$, we update accordingly the relation of $\alpha_{v'}$, and we remove $\alpha_v$ from the "aggregate query" and $v$ from $T$. 
    
    More concretely, if $\alpha_{v'} = R(\bar x)$, we define $(D',\ann')$ to be like $(D,\ann)$ for all relations except for $R$. For the relation $R$ we add,
    for every $(f,n) \in \evalaggD{(\alpha_{v'} \sjoin[\aK] \alpha_v)}{D,\ann}$, the "fact" $R(f(\bar x))$ to $D'$ and define $\ann'(R(f(\bar x))) = n$. By \Cref{lem:semijoin-complexity} this can be done in $\+O(\maxSize{D})$.

    We can now eliminate one atom from $\gamma$, obtaining $\gamma' = \pi_{X_n}^{\oplus_n} \dotsb \pi_{X_1}^{\oplus_1} {\Join^\otimes} (q \setminus \set{\alpha_v})$.
    We finally have
        $\evalaggD{\gamma}{D,\ann} = \evalaggD{\gamma'}{D',\ann'}$.

    We iterate the argument above for $\gamma'$ and $(D',\ann')$, and the result of removing $v'$ from $T$, which satisfies all hypotheses.
    We are finally left with a query $\gamma$ of the form (a) $\gamma = \pi_{X}^{\oplus} {\Join^\otimes} \set{\alpha}$ where $\alpha$ is an "atom", or (b) $\gamma = {\Join^\otimes} q$ for some %
    "full CQ" $q$. Let us see how to deal with these.

    \proofcase{(a)} "Wlog" suppose $\alpha = R(\bar x \bar y)$ and $X = \vars(\bar x)$. Let $A$ be an array of dimension $\dimtup(\bar x)$. For each fact $R(\bar t \, \bar t')$ of $D$ compatible%
            \footnote{That is, $\dimtup(\bar t) = \dimtup(\bar x)$, $\dimtup(\bar t') = \dimtup(\bar y)$, and for every $i,j$, if $(\bar x \bar y) [i] = (\bar x \bar y) [j]$ then $(\bar t  \bar t')[i] = (\bar t \bar t')[j]$.} 
    with $\bar x \bar y$ with "annotation" $\ann(R(\bar t \, \bar t')) = n$, we test if $A[\bar t] = \bot$. If so, we write $A[\bar t] = n$, otherwise we write $A[\bar t] = A[\bar t] \oplus n$. The resulting array $A$ represents the set of "annotation" answers.
    Now we do a second iteration, and for each compatible "facts" $R(\bar t \, \bar t')$ of $D$ we add the pair $(\bar x \assto \bar t,A[\bar t])$ to a linked list $\+L$ (initially empty).
    By iterating the list $\+L$ we can then obtain the "constant delay enumeration@constant delay enumeration after a linear pre-processing" of the output.

    Observe that in the case $X=\emptyset$, $\+L$ will contain just one element and the whole process will be bounded by $\+O(\sizeofq \cdot \maxSize{D})$.

    \proofcase{(b)} Suppose $q = \set{R_1(\bar x_1), \dotsc, R_s(\bar x_s)}$. This is essentially the known algorithm of "constant delay enumeration after a linear pre-processing" for "full CQs" \cite[Theorem~21]{BaganDG07}%
        \footnote{See \cite[\S 4.1]{berkholz2020constant} for a clear and didactic presentation of its proof.} 
    with the addition of $\otimes$-multiplying the "annotations". Such an addition to the algorithm is very easy; in fact, it is not even necessary to revisit the algorithm. 
    It can be implemented by extending each "relation name" $R_\ell$ with an extra column and each "atom" $R_\ell(\bar x_\ell)$ with an extra variable $R_\ell(\bar x_\ell y_\ell)$. We replace "facts" $R_\ell(\bar t)$ in $D$ with $R_\ell(\bar t \, n)$, for $n= \ann(R_\ell(\bar t))$. 
    Now we can use the enumeration for standard "full CQs" as a ``black box'': each time the algorithm outputs an "assignment" $(\bar x_1y_1 \dotsb \bar x_s y_s \assto \bar t_1 n_1 \dotsb \bar t_s n_s)$, we rather output the pair $((\bar x_1 \dotsb \bar x_s \assto \bar t_1 \dotsb \bar t_s), (n_1 \otimes \dotsb \otimes n_s))$.
\end{proof}
We can now simply reduce the general case to the previous one in a rather standard way, by pre-computing the "atoms" on the "bags" and giving them fresh "relation names".
\begin{theoremrep}\AP\label{thm:complexity-aggregate-evaluation}
    Given a "$K$-annotated database" $(D,\ann)$, a "semiring $K$-aggregate query" $\gamma = \pi_{X_n}^{\oplus_n} \dotsb \pi_{X_1}^{\oplus_1} {\Join^\otimes} q$, and a "project-connex" "tree decomposition" of "width" $k$ thereof, the "evaluation" result $\evalaggD{\gamma}{D,\ann}$ admits a "constant delay enumeration after a polynomial 
    $\+O\big(\maxSize{D}^k\big)$ pre-processing".
    If $X_n=\emptyset$, the 
    evaluation 
    is
    in 
    $\+O\big(\sizeofgamma \cdot \maxSize{D}^k\big)$
    in "combined complexity".
\end{theoremrep}
\begin{proof}
    Let $(T,\bagmap,\atommap)$ be a "width" $k$ "project-connex" "tree decomposition"  of $\gamma$. In light of \Cref{lem:bound-decompositions-agg-queries}, we can assume that $|\vertex{T}| \in \+O(\sizeofq)$.
    Let 
    $\atommaplab : \vertex{T} \to \pset{q}$ be an "atom labeling", and  let $f$ be any mapping $f: q \to \vertex{T}$ such that $\vars(\alpha) \subseteq \bagmap(f(\alpha))$ for every $\alpha \in q$ ---it exists due to the "completeness condition".

    Take any vertex $v \in \vertex{T}$ and consider the result of evaluating  $\atommaplab(v)$ (seen as a "full CQ") on $D$, and using this to obtain the set $H_v$ of all pairs $(h|_{\bagmap(v)},n)$ where $h : \atommaplab(v) \homto D$ and 
    $n = \abotimes{\alpha \in f^{-1}(v)}{\ann(h(\alpha))}$.
    This can be obtained in $\+O(\maxSize{D}^k)$. 

    Let $\bar x_v$ be the tuple associated\footnote{"Ie", $\bar x_v$ is any tuple of "variables" without repetitions such that $\vars(\bar x_v)=\bagmap(v)$.} with $\bagmap(v)$.
    Now we create a "$K$-annotated database" having a relation $R_v$ of "arity" $|\bagmap(v)|$. It has a fact $R_v(\bar t)$ with "annotation" $n$ if $\bar t = h(\bar x_v)$ for some $(h,n) \in H_v$. 

    We do the same for all vertices $v \in \vertex{T}$, obtaining a new "$K$-annotated database" $(D',\ann')$ on a "schema" $\set{R_v : v \in \vertex{T}}$.
    Note that this is done in $\+O(|\vertex{T}| \cdot \maxSize{D}^k) \leq \+O(\sizeofgamma \cdot \maxSize{D}^k)$, and that $\maxSize{D'} = \+O(\maxSize{D}^k)$.
    
    Consider $q'$ to be $\set{R_v(\bar x_v) : v \in \vertex{T}}$ in $\gamma$, and $\gamma'$ to be the result of replacing $q$ with $q'$ in $\gamma$.
    \begin{claim}
        $\evalaggD{\gamma}{D,\ann} = \evalaggD{\gamma'}{D',\ann'}$.
    \end{claim}
    \begin{nestedproof}[Proof of Claim]
        We show that $\evalaggD{{\Join^\otimes}q}{D,\ann} = \evalaggD{{\Join^\otimes}q'}{D',\ann'}$, from which the statement follows. %
        
        On the first hand, 
        observe that, for every mapping $h$, we have $h : q \homto D$ "iff" $h: q' \homto D'$.

        On the other hand, we have
        \begin{align*}
            \abotimes{\alpha \in q}{\ann(h(\alpha))}
            = 
            \abotimes{v \in \vertex{T}}{\abotimes{\alpha \in f^{-1}(v)}{\ann(h(\alpha))}} 
            = 
            \abotimes{v \in \vertex{T}}{\ann'(h(R_v(\bar x_v)))} 
            = 
            \abotimes{\alpha \in q'}{\ann'(h(\alpha))}
        \end{align*}
        and hence $(h,n) \in \evalaggD{{\Join^\otimes}q}{D,\ann}$ "iff" $(h,n) \in  \evalaggD{{\Join^\otimes}q'}{D',\ann'}$.
    \end{nestedproof}

    Observe that $(T,\set{ v \mapsto \bagmap(v)}_{v \in \vertex{T}},\set{ v \mapsto \set{\bagmap(v)}}_{v \in \vertex{T}})$ is a "project-connex" "join tree" for $\gamma'$ with a "witnessing bijection" $\set{v \mapsto R_v(\bar x_v) : v \in \vertex{T}}$. We can now invoke \Cref{lem:evalAggQuery:width1} to conclude.
\end{proof}

For $k=1$ the actual decomposition is not needed to grant linear-time evaluation.
\begin{theorem}\AP\label{thm:complexity-aggregate-acyclic-evaluation}
    Given a "$K$-annotated database" $(D,\ann)$, a "semiring $K$-aggregate query" $\gamma=\pi_{X_n}^{\oplus_n} \dotsb \pi_{X_1}^{\oplus_1} {\Join^\otimes} q$ with $\pghw(\gamma)=1$, the "evaluation" result $\evalaggD{\gamma}{D,\ann}$ admits a "constant delay enumeration after a linear pre-processing".
    If $X_n=\emptyset$, the evaluation result can be obtained in 
    $\+O\big(n \cdot \sizeofq + \sizeofq \cdot \maxSize{D}\big)$
    in "combined complexity".
\end{theorem}
\begin{proof}
    We first obtain, in $\+O(n \cdot \sizeofq)$, the needed "project-connex" "tree decomposition" for $\gamma$ by \Cref{cor:lineartime-width1-pc-decompositions}, and we then apply the previous \Cref{thm:complexity-aggregate-evaluation}.
\end{proof}

Finally, we observe that in fact any class of bounded "project-connex" "generalized hyperwidth" has tractable evaluation, even in the absence of decompositions.
\begin{theoremrep}\AP\label{thm:bounded-pghw-ptime}
    For every class $\+C$ of "aggregate queries" such that $\pghw(\+C)$ is bounded, $\+C$ 
    admits a "constant delay enumeration after a polynomial pre-processing".
\end{theoremrep}
\begin{proof}
    Suppose $k$ is a bound for the "project-connex generalized hyperwidth" of $\+C$ and let $\gamma \in \+C$.
    Then we know that $\ghw(\aug{\gamma})$ is at most $k$ by \Cref{thm:width-augmented-char}.
    In turn, by \Cref{lem:ghw:hw:bound} we know  
    that the "hyperwidth" is at most $3k+1$. We can then compute a "hypertree decomposition" for $\aug\gamma$ in polynomial time by \Cref{lem:recognizability:hw:ptime}.
    Finally, from such decomposition we can build, in polynomial time, a "project-connex" "tree decomposition" for $\gamma$ of "width" at most $3k+1$ again by \Cref{thm:width-augmented-char}. With such decomposition at hand, we can now apply \Cref{thm:complexity-aggregate-evaluation} to compute the evaluation of $\gamma$.
\end{proof}

\section{The Case of Counting CQs}
\label{sec:countingcq}

\AP
Given a "CQ" $q(X)$ we define the numerical query $\intro*\counting{q(X)} : \DBs[] \to \Nat$ as the query assigning, for every "database" $D$ the number of elements of $\evalqD{q(X)}D$, denoted by $\intro*\evalCounting{q(X)}{D}$. We call such queries ""counting CQs"".
Given a class $\+C$ of "CQs", let 
$\intro*\counting{\+C} \eqdef \set{\counting{q(X)} : q(X) \in \+C}$.
Counting queries have received considerable attention in the literature.
In general, $\evalPb{\counting{\text{"CQ"}}}$ is "shNP"-complete \cite[Proposition 7]{DBLP:conf/sat/BaulandCCHV04}, and it is "shP"-complete for the restricted class of "full" "CQs", as a generalization of $\#$-SAT ---indeed, it is "shP"-hard even for "acyclic" "CQs" with a single quantified variable \cite[Theorem 4]{PichlerS13}.

A major result known for "counting CQs" is the dichotomy of \cite{DurandMengel15} identifying exactly which classes of queries can be evaluated in polynomial time.  This was later sharpened into a trichotomy in the parameterized complexity setting by \cite{DBLP:conf/icdt/ChenM15}, offering a better understanding of the computational landscape.
Since "counting CQs" are a special form of "aggregate queries", we begin by revisiting the concrete bounds on combined complexity for evaluation established by the simple algorithms discussed earlier. 
We will then show an important observation: the polynomial-time tractability criterion of \cite{DurandMengel15,DBLP:conf/icdt/ChenM15} coincides with the class being of bounded "project-connex" width. Hence, a class $\class$ of "counting CQs" is tractable "iff" $\pghw(\class) < \infty$.

Evaluating $\counting{q(X)}$ on a "database" $D$ is equivalent to evaluating the "semiring $\Nat$-aggregate query" $\pi_{\emptyset}^{+}\pi_X^{\max} {\Join^{\times}} q$ on the "$\Nat$-annotated database" $(D,\oneann)$, where remember $\oneann$ is the $1$-constant function and $\times$, $\max$, $+$ are the multiplication, maximum and sum of natural numbers, respectively.
Observe also that the "project-connex" "generalized hyperwidth" of $\gamma$ is the same as the "free-connex" "generalized hyperwidth" of $q(X)$.
We can therefore apply the upper bounds we showed in previous section,\footnote{These combined complexity bounds are closely related to the data complexity bounds described in \cite[Theorem 4.6]{berkholz2020constant}, but our results are not direct corollaries.}
and further generalize tractability for any semantically bounded "free generalized hyperwidth" class.
\begin{toappendix}
\begin{theorem}[Corollary of \Cref{thm:complexity-aggregate-acyclic-evaluation}]\AP\label{thm:count-complexity-width1}
    Given a "database" $D$, and "free-connex" "acyclic" "counting CQ" $\counting{q(X)}$ ("ie", $\fghw(q(X))=1$), we can compute the "evaluation" $\evalCounting{q(X)}{D}$ in $\+O\big(\sizeofq \cdot \maxSize{D}\big)$.
\end{theorem}
\end{toappendix}
\begin{theoremrep}\AP\label{thm:count-complexity}\label{thm:CQ-bounded-fghw-counting-ptime}
    Given a "database" $D$, a "counting CQ" $\counting{q(X)}$, and a "width" $k$ "free-connex" "tree decomposition"  thereof, we can compute the "evaluation" $\evalCounting{q(X)}{D}$ in $\+O\big(\sizeofq \cdot \maxSize{D}^k\big)$.
    Further, for any class $\+C$ of "CQs" s.t.\ $\fghw(\core(\+C))$ is bounded, $\evalPb{\counting{\+C}}$ is in polynomial time.
\end{theoremrep}
\begin{proof}
    The first statement on combined complexity for evaluation is a direct corollary of \Cref{thm:complexity-aggregate-evaluation}.

    For the second statement, Suppose $k$ is a bound for the "free generalized hyperwidth" of $\core(\+C)$. 
    Let $q'(X) \in \+C$ be such that $\core(q'(X)) = q(X)$; remember that we can compute $q$ from $q'$ in polynomial time due to \Cref{lem:core-sem-bound-ghw-polytime}.
    Then we know that both $\ghw(q(X))$ and $\ghw(\contr(q(X)))$ 
    are at most $k$ by \Cref{lem:bound-free-gen-eq}-(i).
    In turn, by \Cref{lem:ghw:hw:bound} we know  
    that the "hyperwidth" is at most $3k+1$. We can then compute "hypertree decompositions" for $q(X)$ and $\contr(q(X))$ in polynomial time by \Cref{lem:recognizability:hw:ptime}.
    Finally, from these decompositions we can build a "free-connex" "tree decomposition" for $q(X)$ of "width" at most $2\cdot (3 k+1) + 1$ by \Cref{lem:bound-free-gen-eq}. With such decomposition at hand, we can now apply \Cref{thm:count-complexity} above to compute $\counting{q(X)}$ in polynomial time.    
\end{proof}

\AP
We further show that, if we restrict ourselves to the classes $\+C$ of queries with ""bounded arity"" ("ie", there is some $k$ such that all queries of $\+C$ use relations of "arity" at most $k$), then the theorem above is the best one can aim for. That is, classes of unbounded semantic $\fghw$ are not in polynomial time.
This result relies on the characterization of \cite{DBLP:conf/icdt/ChenM15} for the evaluation of $\counting{}$"CQs", which establishes a fine-grained trichotomy in terms of parameterized complexity based on the dichotomy due to Durand and Mengel \cite[Theorem 10]{DurandMengel15}. Concretely, they show that $\evalPb{\counting{\+C}}$ is in polynomial time if, and only if, $\tw(\core(\+C))$ and $\tw(\contr(\+C))$ are both bounded, 
where $\reintro*\contr(\+C)$ is the set of the so-called `"contractions"' of the queries of $\+C$,
which are "hypergraphs" over the free variables (detailed definition in \Cref{sec:chenmengel}).%
We can further show that \cite{DBLP:conf/icdt/ChenM15}'s condition on the class $\+C$ for polynomial time evaluation is equivalent to $\fghw(\core(\+C))$ being bounded, obtaining a dichotomy as corollary.

\begin{toappendix}
\subsection{Chen-Durand-Mengel Characterization}
\label{sec:chenmengel}
In order to state \cite{DBLP:conf/icdt/ChenM15}'s characterization, we need first to introduce some new concepts.
\AP
A ""frontier-path@@CQ""%
    \footnote{\AP\label{ft:frontier-path}In the case of "CQs", a "frontier-path@@CQ" in $q(X)$ is very close to what \cite[Definition 28]{BaganDG07} calls an ``$S$-path'' (and \cite[Definition 5.5]{berkholz2020constant} calls a ``bad path'') in $\hyperq$, with $S = X$. The only difference is that an $S$-path additionally requires that $x$ and $y$ cannot appear in a single hyperedge of $\hyperq$ ("ie", a single "atom" of $q$). 
    It is also equivalent to $x,y$ belonging to the same ``$S$-component'' in the jargon of \cite[p.~115]{DBLP:conf/icdt/ChenM15}; or $x,y$ being adjacent in the ``frontier hypergraph'' in the jargon of \cite[Definition~3.3]{ChenGrecoMengelScarcello23}.} 
in the context of a "CQ" $q(X)$ between two distinct "variables" $x,y$ is a "frontier-path" in the "aggregate query" $\pi^\otimes_X {\Join^\otimes} q$.
\AP
The ""contraction"" of $q(X)$, noted $\intro*\contr(q(X))$, is a "hypergraph" $(V,E)$ where $V$ is the set $X$ of free variables and $E$ contains (i) a hyperedge $e \subseteq V$ if there is an atom $\alpha$ of $q$ whose free variables are $e$ ("ie", $\vars(\alpha) \cap X = e$), and (ii) a hyperedge $\set{x,x'}$
if $x$ and $x'$ are connected via a "frontier-path@@CQ" in $q(X)$. See \Cref{fig:contraction-example} for an example.
\begin{figure}\AP
    \includegraphics[scale=.74]{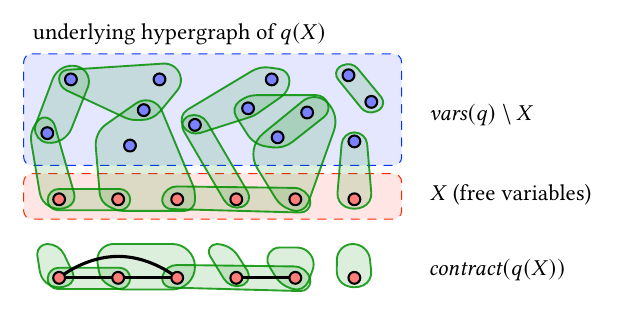}
    \caption{The "contraction" of a "CQ". Black lines represent hyperedges $\set{x,x'}$ derived from the existence of "frontier-paths@@CQ".}
    \label{fig:contraction-example}
\end{figure}
\begin{lemma}\AP\label{lem:tw-contr-core}
    For every "CQ" $q(X)$ we have that $\tw(\contr(q(X))) = \tw(\contr(\core(q(X))))$.
\end{lemma}
\begin{proof}
    This is because every "frontier-path@@CQ" between $x$ and $y$ in $q(X)$ is necessarily witnessed by an atom of $\core(q(X))$ containing both $x$ and $y$.
\end{proof}
For a class $\+C$ of "CQs", let $\reintro*\contr(\+C)$ be $\set{\contr(q(X)) : q(X) \in \+C}$. 

The characterizations of Chen, Durand, and Mengel \cite{DurandMengel15,DBLP:conf/icdt/ChenM15} in particular show the following:
\begin{proposition}\cite[corollary of Theorem 22]{DBLP:conf/icdt/ChenM15}
    \AP\label{thm:chenmengel:charPtime}
    Assuming "FPT" $\neq$ "W1", for every r.e. class $\+C$ of "CQs" of "bounded arity" the following are equivalent:
    \begin{enumerate}
        \item $\evalPb{\counting{\+C}}$ is in polynomial time,
        \item $\tw(\core(\+C))$ and $\tw(\contr(\+C))$ are both bounded.%
    \end{enumerate}
\end{proposition}
\end{toappendix}

\begin{toappendix}
    \subsection{Characterization in Terms of Project-connex Width}
\end{toappendix}
\begin{lemmarep}\AP\label{lem:bound-free-gen-eq-body}
    For any class $\+C$ of "CQs" we have that $\fghw(\+C)$ is bounded if, and only if, both $\ghw(\+C)$ and $\ghw(\contr(\+C))$ are bounded.%
\end{lemmarep}
\begin{proof}
    We will rather show the more fine-grained version of the statement given below as \Cref{lem:bound-free-gen-eq}, with concrete bounds on the grow of decompositions.
\end{proof}
\begin{toappendix}
\begin{lemma}\AP\label{lem:bound-free-gen-eq}
    For any class $\+C$ of "CQs" the following are equivalent:
    \begin{enumerate}
        \item \AP\label{lem:bound-free-gen-eq:1}  $\fghw(\+C)$ is bounded,
        \item \AP\label{lem:bound-free-gen-eq:2} $\ghw(\+C)$ and $\ghw(\contr(\+C))$ are both bounded.
    \end{enumerate}
    Further, 
    (i) if $\fghw(\+C) \leq k$ then  $\ghw(\contr(\+C))\leq k$ and $\ghw(\+C) \leq k$, and
    (ii) if $\ghw(\+C) \leq k$ and $\ghw(\contr(\+C)) \leq k$, then $\fghw(\+C) \leq 2 k$.
    Witnessing decompositions of these inequalities can be obtained from one another in polynomial time.
\end{lemma}

\begin{proof}
    \proofcase{\eqref{lem:bound-free-gen-eq:1} $\Rightarrow$ \eqref{lem:bound-free-gen-eq:2}}
    Let $q(X) \in \+C$ have a "free-connex" "tree decomposition" $(T,\bagmap, \atommap)$ of "width" $k$. It follows that $q(X)$ has "generalized hyperwidth" ${\leq} k$. 
    We are left to show that $\contr(q(X))$ has "generalized hyperwidth" ${\leq} k$.
    For this, consider $T'$ to be the "witness subtree for $X$" of $T$,
    and let $\bagmap'$ be the restriction of $\bagmap$ onto $\vertex {T'}$.
    \begin{claim}
        $(T',\bagmap')$ is a "tree decomposition" for $\contr(q(X))$.
    \end{claim}
    \begin{nestedproof}
        To show the claim, take any edge $\set{x,x'}$ of $\contr(q(X))$. 

            If $x,x'$ are both part of an "atom" of $q$, then they are present in an edge of $\hyperq$, and then $\set{x,x'} \subseteq \bagmap(v)$ for some vertex $v$ of $T$ by the "completeness condition" of $(T,\bagmap,\atommap)$. It follows that $\set{x,x'}$ must belong to $\bagmap'(v')$ where $v'$ is the vertex of $T'$ closest to $v$.

            Otherwise, there is a "frontier-path@@CQ" between $x$ and $x'$ of length $>1$. Let $v_x$ and $v_{x'}$ be vertices of $T'$ such that $x \in \bagmap'(v_x)$ and $x' \in \bagmap'(v_{x'})$.
            By the "connectivity condition" of $(T,\bagmap,\atommap)$, there must be a path in $T$ from $v_{x}$ to $v_{x'}$ in $T$ which can be divided into three parts $P_1 P_2 P_3$, where $P_1$ uses "bags" of $T'$ "containing@@bag" $x$, $P_3$ uses "bags" of $T'$ "containing@@bag" $x'$, and $P_2$ uses only 
            "bags" from $\vertex{T} \setminus \vertex{T'}$ witnessing the "frontier-path@@CQ". Since $T$ is a tree, $P_2$ must start and end on the same "bag" $v$ of $T'$, witnessing that $\set{x,x'} \subseteq \bagmap(v)$. This concludes the proof of the claim.
    \end{nestedproof}
    
    Thus, both the "generalized hyperwidth" of $q(X)$ and $\contr(q(X))$ are bounded by $k$.

    \proofcase{\eqref{lem:bound-free-gen-eq:2} $\Rightarrow$ \eqref{lem:bound-free-gen-eq:1}}
    Let $k$ be the bound on the "generalized hyperwidth" of $q(X)$ and of $\contr(q(X))$ for every $q(X) \in \+C$. We will show that there exists a "free-connex" "tree decomposition" of "width" ${\leq} 2k+1$ for $q(X)$. 
    Please refer to the visual aid of Figure~\ref{lem:fig:bound-free-gen-eq} for understanding of the proof that follows.
    \begin{figure*}
        \newcommand{\scalefactor}{0.74}
        \centering
        \begin{subfigure}{.35\linewidth}
            \centering
            \includegraphics[scale=\scalefactor]{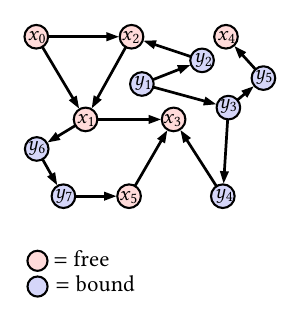}
            \caption{
                \AP\label{lem:fig:bound-free-gen-eq:q}
                An example query $q(X)$.
            }
        \end{subfigure} 
        \hfill 
        \begin{subfigure}{.3\linewidth}
            \centering
            \includegraphics[scale=\scalefactor]{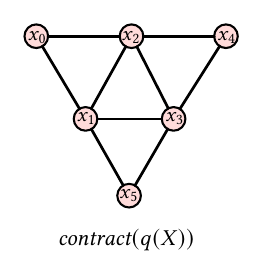}
            \caption{
                \AP\label{lem:fig:bound-free-gen-eq:contr}
                The "contraction" of $q(X)$.
            }
        \end{subfigure}
        \hfill
        \begin{subfigure}{.3\linewidth}
            \centering
            \includegraphics[scale=\scalefactor]{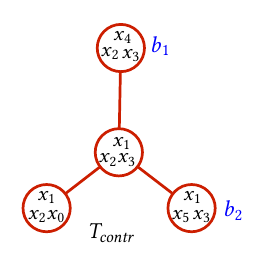}
            \caption{
                \AP\label{lem:fig:bound-free-gen-eq:Tcontr}
                The "tree decomposition" of the contraction of $q(X)$.
            }
        \end{subfigure}
        \\
        \begin{subfigure}{.45\linewidth}
            \centering
            \includegraphics[scale=\scalefactor]{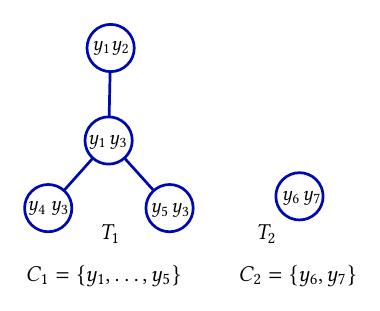}
            \caption{
                \AP\label{lem:fig:bound-free-gen-eq:Ti}
                The "tree decomposition" of each component $C_i$.
            }
        \end{subfigure} 
        \hfill
        \begin{subfigure}{.45\linewidth}
            \centering
            \includegraphics[scale=\scalefactor]{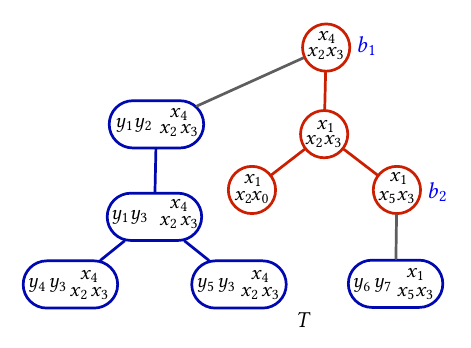}
            \caption{
                \AP\label{lem:fig:bound-free-gen-eq:T}
                The "free-connex" "tree decomposition" of $q(X)$.
            }
        \end{subfigure}
        \caption{
            \AP\label{lem:fig:bound-free-gen-eq}
            Visual aid showing, for an example query \eqref{lem:fig:bound-free-gen-eq:q} (for simplicity, every "relation name" is binary, and $q$ is hence depicted as an edge-labeled directed graph), the construction of a "free-connex" "tree decomposition" \eqref{lem:fig:bound-free-gen-eq:T} from the "tree decompositions" for each bounded component $C_i$ \eqref{lem:fig:bound-free-gen-eq:Ti} and the "tree decomposition" \eqref{lem:fig:bound-free-gen-eq:Tcontr} of its "contraction" \eqref{lem:fig:bound-free-gen-eq:contr}.
        }
    \end{figure*}

    Let $B = \vars(q) \setminus X$ be the set of "bound variables" of $q(X)$.
    Consider the "hypergraph" $\restrictG[\hyperq]{B}$ induced by the "bound variables" of $q$ on its "underlying hypergraph", and suppose it is partitioned into $n$ connected components $C_1 \dcup \dotsb \dcup C_n = B$. 
    Observe that each connected component $\restrictG[\hyperq]{C_i}$ has "generalized hyperwidth" ${\leq} k$, since the "generalized hyperwidth" of $\hyperq$ is ${\leq} k$ and "generalized hyperwidth" is closed under taking induced sub-"hypergraphs". Let then $(T_i,\bagmap_i,\atommap_i)$ be a "tree decomposition" for $\restrictG[\hyperq]{C_i}$ of "width" ${\leq} k$, for each $i \in [n]$.
    
    For every $i \in [n]$, let $F_i \subseteq X$ be the set of all "free variables" incident to $C_i$ in $\hyperq$ ("ie", $F_i \eqdef \set{x \in X : \set{x,x'} \subseteq e \text{ for some } e \in \edges{\hyperq} \text{ and } x' \in C_i}$).
    Notice that, by definition, there is a hyperedge in $\contr(q)$ containing $x$ and $x'$ if{f} either 
    \begin{enumerate}
        \item some "atom" of $q$ contains both $x$ and $x'$, or
        \item $\set{x,x'} \subseteq F_i$ for some $i \in [n]$.
    \end{enumerate}

    Take a "tree decomposition" $(T_{contr},\bagmap_{contr},\atommap_{contr})$ of $\contr(q)$ of "width" ${\leq} k$. By the observation above, for every $i \in [n]$, we have that $F_i$ forms a clique in $\contr(q)$, and therefore there must be a "bag" $b_i$ of $T_{contr}$ "containing@@bag" $F_i$ due to \Cref{lem:cliques-in-bags}. 
    In particular, this means that $|\atommap(b_i)| \leq k$ since $(T_{contr},\bagmap_{contr},\atommap_{contr})$ has "width" ${\leq} k$.
    For every $i \in [n]$, let $T'_i$ be the result of adding $F_i$ to  $\bagmap(v)$ and adding $\atommap(b_i)$ to $\atommap(v)$, for every "bag" $v$ of $T_i$: in the figure, starting with the two $T_i$ decompositions of \Cref{lem:fig:bound-free-gen-eq:Ti} we obtain the two blue $T'_i$ decompositions of \Cref{lem:fig:bound-free-gen-eq:T}. This results in a "tree decomposition" for $\restrictG[\hyperq]{C_i \dcup F_i}$ of "width" ${\leq} 2k$.\footnote{We certainly abuse notation when calling $T'_i$ a ``"tree decomposition" for $\restrictG[\hyperq]{C_i \dcup F_i}$'' since, strictly speaking, $\atommap(b_i)$ may well contain hyperedges with vertices outside $C_i \dcup F_i$.}
    Finally, let $T$ be the result of attaching $T'_i$ as a child of $b_i$ in $T_{contr}$, for every $i \in [n]$. 
    That is, for each $i$ we add $T'_i$ and create an edge between $b_i$ and a node (any node) of $T'_i$. See how \Cref{lem:fig:bound-free-gen-eq:T} is built from \Cref{lem:fig:bound-free-gen-eq:Ti} and \Cref{lem:fig:bound-free-gen-eq:Tcontr}.
    It follows that $T$ is indeed a "free-connex" "tree decomposition" for $q(X)$ whose every "bag" has size ${\leq} 2k$.%
\end{proof}
\end{toappendix}
\begin{corollaryrep}[of \Cref{lem:bound-free-gen-eq-body} and {\cite[Theorem 22]{DBLP:conf/icdt/ChenM15}}]\AP\label{cor:freeghw:charPtime-body}
    If "FPT" $\neq$ "W1", for every "re" class $\+C$ of "CQs" of "bounded arity" we have:
        $\evalPb{\counting{\+C}}$ is in "PTime"  $\Leftrightarrow$  $\fghw(\core(\+C))$ is bounded.
\end{corollaryrep}
\begin{proof}
    Instead of proving \Cref{cor:freeghw:charPtime-body} we show a slightly more fine-grained \Cref{cor:freeghw:charPtime}, which can be found below.
\end{proof}

\begin{toappendix}
\begin{corollary}\AP\label{cor:freeghw:charPtime}
    Assuming "FPT" $\neq$ "W1", for every "re" class $\+C$ of "CQs" of "bounded arity" the following are equivalent:
    \begin{enumerate}
        \item\AP\label{cor:freeghw:charPtime:1} $\evalPb{\counting{\+C}}$ is in polynomial time,
        \item\AP\label{cor:freeghw:charPtime:2} $\fghw(\core(\+C))$ is bounded,
        \item\AP\label{cor:freeghw:charPtime:3} $\ghw(\core(\+C))$ and $\ghw(\contr(\+C))$ are bounded,
        \item\AP\label{cor:freeghw:charPtime:4} $\tw(\core(\+C))$ and $\tw(\contr(\+C))$ are bounded,
        \item \AP\label{cor:freeghw:charPtime:5} for some $k \in \Nat$, every "CQ" in $\+C$ is "equivalent" to a "CQ" of $\fghw$ at most $k$.
    \end{enumerate}
\end{corollary}
\begin{proof}
    \eqref{cor:freeghw:charPtime:1} $\Leftrightarrow$ \eqref{cor:freeghw:charPtime:4} follows from \Cref{thm:chenmengel:charPtime};  \eqref{cor:freeghw:charPtime:2} $\Leftrightarrow$ \eqref{cor:freeghw:charPtime:3} follows from \Cref{lem:bound-free-gen-eq} combined with \Cref{lem:tw-contr-core}; \eqref{cor:freeghw:charPtime:3} $\Leftrightarrow$ \eqref{cor:freeghw:charPtime:4} follows from \Cref{rk:ghw-leq-tw}; and
    \eqref{cor:freeghw:charPtime:2} $\Leftrightarrow$ \eqref{cor:freeghw:charPtime:5} follows from \Cref{lem:semantic-fghw}.
\end{proof}
\end{toappendix}

\section{Concluding Remarks}\label{sec:conclusion}

\begin{toappendix}
    \label{app:conclusion}
\end{toappendix}

We have introduced and studied the "project-connex" width on "aggregate queries" and proposed it as an alternative structural measure which enjoys several desirable properties, namely: 
(i) having a simple intuitive definition, 
(ii) admitting an efficient computation of decompositions, 
(iii) facilitating constant-delay enumeration results, 
(iv) enabling a uniform and simple algorithmic treatment for evaluation.
A natural extension to this work is adapting our results to other monotone
width measures, such as fractional \cite{GroheM14} or submodular \cite{Marx13} width, adding unions and allowing constants (more details in \Cref{sec:discussion}).
We may also consider if the project-connex generalization of free-connex fractional hypertree width (\textit{fn-fhtw}) can lead to results on output-optimal algorithms for evaluating aggregate queries comparable to the ones showed in \cite{Hu25} for conjunctive queries.

While we have compared our approach to prior work throughout the manuscript, we would like to emphasize that our work is very much aligned with seminal work found in \cite{KhamisNR16,JoglekarPR16,khamis2023faqquestionsaskedfrequently} on aggregate query languages. However, the focus of these works is put on the semantic equivalence between queries, 
while ours is purely structural.

As mentioned in \Cref{sec:countingcq}, 
previous characterization criteria for counting conjunctive queries have been stated in terms of `quantified star-size' \cite{DurandMengel15} or the tree-width measure on two queries \cite{DBLP:conf/icdt/ChenM15,ChenGrecoMengelScarcello23}. 
As we showed, if the focus is solely on a tractability characterization, considering the notion of width on "free-connex" tree decompositions is enough (however, this does not translate into a simplification for the \emph{trichotomy} characterization statement of \cite{DBLP:conf/icdt/ChenM15}).

\begin{toappendix}

\label{sec:discussion}

We discuss some natural extensions and generalizations.

\subparagraph*{Fractional, Submodular, and Other Monotonic Widths} 
While we tried to minimize the number of definitions regarding tree decompositions, we remark that our results can be easily adapted to other monotone\footnote{Monotone in the sense that removing vertices from "bags" can only decrease the width measure.}  width measures, such as fractional \cite{GroheM14} or submodular \cite{Marx13} width. For example, an analog statement to that of \Cref{thm:width-augmented-char} can be made for computing ``"project-connex" fractional width''; and the bounds of \Cref{thm:complexity-aggregate-evaluation} for evaluation can be trivially adapted by replacing the exponent $k$ with the "project-connex" fractional width. %

    A \AP""measure"" is a function $\mu$ that takes a "hypergraph" and a set of vertices thereof and returns a real number. It is a ""monotone measure"" if further for every "hypergraph" $G$ and $A \subseteq B \subseteq \vertex{G}$ we have $\mu_G(A) \leq \mu_G(B)$. The ""$\mu$-width"" of $G$ is then the minimum value $\max_{v \in T} \mu_G(\bagmap(v))$ among all "tree decompositions" $(T,\bagmap,\atommap)$ of $G$ (in this case $\atommap$ plays no role). This notion of monotone measure is similar to $\gamma$-width, where $\gamma$ is a node-monotone function \cite[Definition 31]{JoglekarPR16}, and $g$-width where $g$ is a monotone function \cite[Definition 46]{khamis2023faqquestionsaskedfrequently}.

    \AP The measures of "tree-width", "generalized hyperwidth", ""fractional hyperwidth"" \cite{GroheM14} and ""submodular width"" \cite{Marx13} can be seen all as instances of "monotone measures".

    \Cref{thm:width-augmented-char} can be restated in this more general setting as follows:
    \begin{theorem}[Generalization of \Cref{thm:width-augmented-char}]\label{thm:width-augmented-char-monotonewidth}
        For every "aggregate query" $\gamma$ and "monotone measure" $\mu$, we have that the "$\mu$-width" of $\gamma$ and of $\aug\gamma$ coincide. 
        Further, decompositions can be obtained from one another in linear time.
    \end{theorem}
    \begin{proof}
        The proof follows exactly the same lines as the proof of \Cref{thm:width-augmented-char}. 
        
        The part of $\pghw(\gamma) \geq \ghw(\aug\gamma)$ uses the same decomposition, hence the "$\mu$-width" is preserved.
        
        For the second part showing $\pghw(\gamma) \leq \ghw(\aug\gamma)$, the base case is identical. For the inductive case, note that the "relativization" of bags cannot increase the "$\mu$-width" due to "measure monotonicity", and hence \Cref{cl:characterization-augmented-preserves-width} can be replaced with a claim which says that the "$\mu$-width" of $(\widetilde T, \widetilde\bagmap, \widetilde\atommap )$ is at most the "$\mu$-width" of $(T,\bagmap,\atommap)$.
        Also, observe that \Cref{lem:cliques-in-bags} holds independently of the width "measure" at hand.
    \end{proof}

    For the evaluation procedure, note that the incidence of the "$\mu$-width" is the first step of the pre-computation of bags in \Cref{thm:complexity-aggregate-evaluation-frac}. Hence, by a trivial adaptation for the case of "$\mu$-width" being "fractional hyperwidth" we can obtain for instance the following:
    \begin{theorem}[Adaptation of \Cref{thm:complexity-aggregate-evaluation}]\AP\label{thm:complexity-aggregate-evaluation-frac}
        Given a "$K$-annotated database" $(D,\ann)$, 
        a "semiring $K$-aggregate query" $\gamma = \pi_{X_n}^{\oplus_n} \dotsb \pi_{X_1}^{\oplus_1} {\Join^\otimes} q$, and 
        a "project-connex" "tree decomposition" of "fractional hyperwidth" $k$ thereof, the "evaluation" result $\evalaggD{\gamma}{D,\ann}$ admits a "constant delay enumeration after a polynomial $\+O\big(\maxSize{D}^k\big)$ pre-processing".
        If $X_n=\emptyset$, the evaluation result can be obtained in 
        $\+O\big(\sizeofgamma \cdot \maxSize{D}^k\big)$
        in "combined complexity".
    \end{theorem}

\subparagraph*{Unions}
\label{sec:unions}
\AP
While we left out the unions for simplicity, they could be added. We say that two "$K$-aggregate queries" 
$\gamma_1 = \pi_{X_1}^{\oplus_1} \hat \gamma_1$ and 
$\gamma_2 = \pi_{X_2}^{\oplus_2} \hat \gamma_2$ 
are ""compatible queries"" if $X_1=X_2$. %
The evaluation of the union $\gamma \cup \gamma'$ of two "$K$-aggregate" "compatible queries" is then defined on a "$K$-annotated database" $D$ as 
$\reintro*\evalaggD{(\gamma \cup^{\oplus} \gamma')}{D,\ann} \eqdef \evalaggD{\gamma}{D,\ann} \cup \evalaggD{\gamma'}{D,\ann}$.
Width parameters are lifted by taking the maximum: $\reintro*\pghw(\gamma \cup^\oplus \gamma') \eqdef \max(\pghw(\gamma),\pghw(\gamma'))$. The upper bounds of \Cref{sec:evalAgg} can be then adapted in a trivial way to unions of "compatible" "semiring $K$-aggregate queries".
Observe, however, that this trivial extension holds for having the union as outermost operator. In particular it \emph{does not cover} \AP$\counting{\text{"UCQ"}}$, where ""UCQ"" refer to finite unions of "CQs", which is a more challenging counting task \cite{ChenM16-UCQ}, and corresponds rather to "aggregate queries" of the form $\pi_{\emptyset}^{+}((\pi_X^{\max}{\Join^{\otimes}}q_1) \cup \dotsb \cup (\pi_X^{\max}{\Join^{\otimes}}q_\ell))$. Indeed, such queries are known to be "shP"-hard even in the absence of $\pi^{\max}$-projections and assuming "acyclicity" ("ie", for counting solutions of unions of "acyclic" "full" "CQs") \cite[Theorem 6]{PichlerS13}.

\subparagraph*{Constants}
\label{sec:constants}
Observe that in our technical developments we preferred to avoid the complications of using "constants" explicitly. We remark that one can handle constants in a rather standard way, which we simply show on an example.
Consider an "aggregate query" $\gamma=\pi_{X_1}^{\oplus_1} \dotsb \pi_{X_n}^{\oplus_n}{\Join^\otimes}q$ possibly with constants and an "annotated database" $(D,\ann)$, and suppose it contains an "atom" $R(x,y,c,x,c')$ where $c,c'$ are constants and $x,y$ are variables. We replace this atom with $R_{(*,*,c,*,c')}(x,y,x)$ and create a new ternary relation $R_{(*,*,c,*,c')}$ in $D$ containing every "fact" $R_{(*,*,c,*,c')}(c_1,c_2,c_3)$ annotated with $n$ satisfying that $R(c_1,c_2,c,c_3,c')$ is a  "fact" of $D$ with annotation $n=\ann(R(c_1,c_2,c,c_3,c'))$. We proceed similarly for all "constants". Finally, we obtain some "annotated database" $(D',\ann')$ and some "aggregate query" $\gamma'=\pi_{X_1}^{\oplus_1} \dotsb \pi_{X_n}^{\oplus_n}{\Join^\otimes}q'$ without "constants" such that
$\evalaggD{\gamma}{D,\ann} = \evalqD{\gamma'}{D',\ann'}$. Further: (i) such $q'$ and $(D',\ann')$ can be computed in linear time, more precisely in $\+O(\sizeofgamma + \sizeofD)$, and (ii)
the "project-connex generalized hyperwidth" of $\gamma'$ is that of $\gamma$.

\begin{toappendix}

\end{toappendix}

\end{toappendix}

\bibliography{long,bib}

\appendix
\end{document}